\documentclass[12pt,a4paper]{amsart}

\usepackage{latexsym,amssymb,amsmath,amsthm,amsfonts,enumerate,verbatim,xspace,
exscale}
\usepackage{graphicx}
\usepackage{color,amsbsy}
\usepackage[all]{xy}

\usepackage{hyperref}
\definecolor{darkblue}{rgb}{0,0,.5}
\hypersetup{colorlinks=true, breaklinks=true, linkcolor=darkblue, menucolor=darkblue,  urlcolor=darkblue,citecolor=darkblue}

\usepackage{setspace}

\theoremstyle{plain}
\newtheorem{theorem}{Theorem}[section]

\newtheorem{proposition}[theorem]{Proposition}

\theoremstyle{definition}
\newtheorem{definition}[theorem]{Definition}
\newtheorem{remark}[theorem]{Remark}

\def\R{\mathbb{R}}
\def\I{\mathbb{I}}

\def\M{\mathcal{M}}

\def\F{\mathcal{F}}

\def\A{\mathcal{A}}
\def\W{\mathcal{W}}

\newcommand{\hor}{\mbox{$\text{\up{Hor}}$}}
\newcommand{\hpr}{\mbox{$\text{\up{hpr}}$}}
\newcommand{\ver}{\mbox{$\text{\up{Ver}}$}}
\newcommand{\vpr}{\mbox{$\text{\up{vpr}}$}}
\newcommand{\hl}{\mbox{$\text{\up{hl}}$}}
\newcommand{\hlA}{\mbox{$\textup{hl}^{\mathcal{A}_0}$}}
\newcommand{\vl}{\mbox{$\text{\up{vl}}$}}
\newcommand{\CurvA}{\mbox{$\textup{Curv}^{\mathcal{A}_0}$}}
\newcommand{\Curv}{\mbox{$\textup{Curv}$}}

\newcommand{\up}{\upshape}

\newcommand{\hookto}{\hookrightarrow}

\def\vv<#1>{\langle#1\rangle}
\def\ww<#1>{\langle\langle#1\rangle\rangle}

\newcommand{\id}{\mbox{$\text{\up{id}}\,$}}
\newcommand{\pr}{\mbox{$\text{\up{pr}}$}}

\newcommand{\dd}[2]{\mbox{$\frac{\partial #2}{\partial #1}$}}

\newcommand{\om}{\omega}
\newcommand{\Om}{\Omega}
\newcommand{\var}{\varphi}

\newcommand{\lam}{\lambda}
\newcommand{\Lam}{\Lambda}

\newcommand{\by}[2]{\mbox{$\frac{#1}{#2}$}}

\providecommand{\set}[1]{\mbox{$\{#1\}$}}

\newcommand{\ao}{\mathfrak{a}}

\newcommand{\doo}{\mathfrak{d}}

\newcommand{\gu}{\mathfrak{g}}

\newcommand{\ko}{\mathfrak{k}}

\newcommand{\mo}{\mathfrak{m}}

\newcommand{\po}{\mathfrak{p}}

\newcommand{\Ad}{\mbox{$\text{\upshape{Ad}}$}}
\newcommand{\ad}{\mbox{$\text{\upshape{ad}}$}}

\newcommand{\SO}{\mbox{$\textup{SO}$}}
\newcommand{\so}{\mbox{$\mathfrak{so}$}}

\newcommand{\X}{\mbox{$\mathcal{X}$}}

\newcommand{\Aut}{\mbox{$\textup{Aut}$}}

\newcommand{\Diff}{\mbox{$\textup{Diff}$}}

\title[Symmetry actuated  Hamiltonian systems]{%
Symmetry actuated closed-loop Hamiltonian systems}

\usepackage[foot]{amsaddr}  
\author{Simon Hochgerner}
\address{\"Osterreichische Finanzmarktaufsicht (FMA),
Otto-Wagner Platz 5, A-1090 Vienna
}
\email{simon.hochgerner@fma.gv.at} 

\begin{document}

\begin{abstract}
This paper extends the theory of controlled Hamiltonian systems with symmetries due to \cite{Kri85,BKMS92,BMS97, BLM01a, BLM01b, CBLMW02} to the case of non-abelian symmetry groups $G$ and semi-direct product configuration spaces. The notion of symmetry actuating forces is introduced and it is shown, that Hamiltonian systems subject to such forces permit a conservation law, which arises as a controlled  perturbation of the $G$-momentum map.  Necessary and sufficient matching conditions are given to relate the closed-loop dynamics, associated to the forced Hamiltonian system, to an unforced Hamiltonian system. These matching conditions are then applied to general Lie-Poisson systems, to the example of ideal charged fluids in the presence of an external magnetic field (\cite{H19}), and to the satellite with a rotor example (\cite{BKMS92,BMS97}).
\end{abstract}

\maketitle


\section*{Introduction} 
\subsection*{Feedback control of Hamiltonian systems with symmetries}
The method of controlled Lagrangian and Hamiltonian systems has started with \cite{Kri85,BKMS92} and then been further developed in \cite{BLM01a, BLM01b, CBLMW02, BMS97, PB19}. Reviews are contained in \cite{BL02} and \cite{MarsdenPCL}. In \cite{CBLMW02} it is shown, for general systems and without explicitly considering symmetries, that the methods of controlled Lagrangian and Hamiltonian systems are equivalent.

The prototypical example with regard to this method is the satellite with a rotor. In this example the satellite is modelled as a rigid body and a rotor is attached to the third principal axis. If the rotor is turned sufficiently fast, the satellite becomes stable for rotations about the middle axis. Since the middle axis is otherwise an unstable equilibrium of the rigid body, this means that the rotor control can be used to stabilize the system. What ``sufficiently fast'' means has been worked out by \cite{BKMS92}. Their construction is based on the observation that the rotor control depends on the satellite body angular momentum, which leads to a feedback system, and this feedback system can be reformulated as a Hamiltonian system. Thus, the (non-linear) stability analysis of \cite{A66,HMRW85} can be used to find a stabilizing control. 

In \cite{BLM01a, BLM01b, CBLMW02} this approach has been systematized in the following manner, both for Lagrangian and Hamiltonian systems: 
The authors of loc.\ cit.\ start with a mechanical system which is assumed to be of Kaluza-Klein type on a configuration space $P$. That is, $\pi: P\to M$ is a (finite-dimensional) principal fiber bundle with structure group $G$; $(\mu_0^M,\I_0,\A_0)$ is a set of Kaluza-Klein data consisting of metric, inertia tensor and connection form, 
such that the metric $\mu_0^P$ on $P$ is given by 
$(\mu_0^P)_{\om}(\xi,\zeta) 
= 
(\mu_0^M)_{\pi(\om)}(T_{\om}\pi.\xi,T_{\om}\pi.\zeta) 
+
(\I_0)_{\om}(\A_0(\om).\xi,\A_0(\om).\zeta)$; the Hamiltonian is of the form $H_0(\Pi) = \by{1}{2}\vv<\Pi,(\mu_0^P)^{-1}\Pi> + V(\om)$ where $\om = \tau(\Pi)$ is the base point and $V: P\to\R$ is a $G$-invariant potential. Now, the Kaluza-Klein data can be perturbed to obtain a controlled set $(\mu_C^M,\I_C,\A_C)$, and thus a controlled Hamiltonian $H_C$ (with the same potential $V$). 
In loc.\ cit.\ it is shown how, under certain conditions, the equations of motion corresponding to $H_C$ are equivalent to the closed-loop dynamics corresponding to $H_0$ subject to a feedback control force.

Further, it is shown in \cite{BLM01a, BLM01b, BL02, BKMS92, BMS97} that this force is constructed so that it acts only along the symmetry directions of $G$ -- e.g., \cite[Equ.~(2.6)]{BLM01a}. 

In order for this construction to work, \cite{BLM01a, BLM01b, BL02, BKMS92, BMS97, PB19} make two types of  assumptions:
\begin{enumerate}
    \item 
    A set of matching assumptions.
    \item
    The symmetry group $G$ is abelian. 
\end{enumerate}
The matching assumptions ensure that the Hamiltonian equations of motion for $H_C$ are equivalent (can be ``matched'') to the desired closed loop equations. The second assumption seems to be of technical nature. Further, \cite{BLM01b} contains an Euler-Poincar\'e version of these results, where it is assumed that $P=M\times G$ is a direct product of Lie groups $M$  and $G$.

\subsection*{Description of results}
In this paper the matching approach is reversed: The starting point is not the perturbation of the set of Kaluza-Klein data, but the control force. Thus, given a $G$-invariant Hamiltonian $H$ and a $G$-equivariant fiber-linear map $j: T^*P\to\gu^*$, we specify in Definition~\ref{def:SymCloLoop} the \emph{symmetry actuating force} $F$. This force acts only along symmetry directions in the sense (Definition~\ref{def:admissible}) that it is a vector field, which is vertical for the \emph{control connection} $\Gamma$ introduced in Section~\ref{sec:Gamma}. In the satellite example the control connection corresponds to the satellite variables, so that work is only done in the rotor direction. 

The equation of motion associated to the Hamiltonian vector field $X_{H}$ and the feedback force $F$ is 
\begin{equation}
    \tag{\ref{e:Feom1}}
    \dot{\Pi}_t
    = X_{H}(\Pi_t) + F(\Pi_t)
\end{equation}
and in Theorem~\ref{thm:ConsLaw} it is shown that a controlled conservation law holds:
\[
 (J+j)(\Pi_t)
\]
is constant in $t$ where $J: T^*P\to\gu^*$ is the momentum map associated to the cotangent lifted action of $G$.  
This constant of motion provides the link to the closed-loop equations of motion by eliminating the $G$-dependence. 

The crucial step in the construction of $F$ is now prescription~\eqref{e:B-j-tilde}, which determines $j$ by means of a $G$-equivariant fiber-linear map $\tilde{\jmath}: \hor^* = J^{-1}(0)\to\gu^*$. This prescription is such that the force depends on the full $\Pi$-dynamics, and not only on the shape directions, which correspond through the choice of the connection $\Gamma$ to positions and momenta of the dynamics in the shape space $M$. (However, actuation is still only on symmetry directions.) For the satellite example, this means that the force depends on the satellite and rotor variables, and not only on satellite observations. This approach is necessary in order to deal with non-abelian symmetry groups $G$ or semi-direct products $M\circledS G$. Also, this approach is different from \cite{BLM01a, BLM01b, BL02, BKMS92, BMS97, PB19} where the perturbation $(\mu_C^M,\I_C,\A_C)$ is constructed so that the corresponding force only depends on the shape directions. 
See also the discussion in Sections~\ref{sec:CON-B=1} and \ref{sec:CON-COMP}. 
Thus, while starting from the force $F$ is in principle equivalent to starting from the perturbed data $(\mu_C^M,\I_C,\A_C)$, the former has, in conjunction with Theorem~\ref{thm:ConsLaw}, the advantage of clarifying the relation between $(\mu_C^M,\I_C,\A_C)$ and the closed-loop equations.  Compare Remarks~\ref{rem:SATrem2} and \ref{rem:SATrem3}.

Theorem~\ref{thm:MC} provides a matching result,  i.e.\ equivalence of closed-loop and Hamiltonian $H_C$-dynamics, under very general assumptions. In particular, $G$ is allowed to be non-abelian.  
This matching theorem is necessary and sufficient, and provides explicit formulas for $(\mu_C^M,\I_C,\A_C)$.  But there is also a fourth condition \eqref{e:thm_MC_2}, which is not explicit. In Theorem~\ref{thm:SD-matching}, this fourth condition is made concrete for the case of a semi-direct product $M\circledS G$ and a right invariant Hamiltonian $H_0$. This condition is
\begin{equation}
    \tag{\ref{e:SD-MCond}}
    A_0u\diamond (1+CA_0^*)^{-1}C\nu 
    = 0
\end{equation}
which depends on the semi-direct product structure. 
Here $A_0: \mo\to\gu$ is the downstairs connection form corresponding to $\A_0$ defined \eqref{e:SD-MC}, and the linear map $C: \mo^*\to\gu^*$, introduced in \eqref{e:SD-C}, determines the fiber-linear map $\tilde{\jmath}: \hor^*\to\gu^*$.
As a corollary, it follows that matching is always possible for direct product Lie groups $M\times G$. 

Formulas~\eqref{e:SD-mu_C^M} and \eqref{e:SD-A_C} in Theorem~\ref{thm:SD-matching} have also been found in \cite{H19}. But now these formulas follow from  a general construction and not from an ad-hoc analysis of the closed-loop equations and the feedback force. In particular, Theorem~\ref{thm:MC} implies that this form of the controlled data is not only sufficient for matching, but also necessary.  Moreover, the construction is now no longer constrained to the realm of Lie-Poisson equations, but can be applied to general Hamiltonian systems with symmetries. When these systems are sufficiently explicit, it may be hoped that concrete matching conditions along the lines of Theorem~\ref{thm:SD-matching} can be found. See Sections~\ref{sec:CON-EMF} (electromagnetic flow control) and \ref{sec:CON-SHC} (stochastic Hamiltonian systems) for possible future directions.  

To illustrate these results, two examples are treated: ideal flow of charged particles subject to an external magnetic field (\cite{H19}) and the satellite with a rotor (\cite{BKMS92}). These examples are not new, but can now be treated both within the same framework and without any ad-hoc assumptions. 
This is also in contrast to \cite{H19}, where these examples were considered together, but yet with respect to different feedback forces and these forces were found through an ad-hoc analysis of the desired closed loop equations. 

Thus this paper is an attempt to synthesize the symmetry control approach of \cite{BLM01a}, restricted to abelian $G$, and the results of \cite{H19}, which are restricted to the Lie-Poisson case. 
Comparing the general matching conditions of \cite{BLM01a} and Theorem~\ref{thm:MC}, it should be mentioned that the conditions $M2$ and $M3$ in Theorem~2.2 of \cite{BLM01a} are more explicit than \eqref{e:thm_MC_2} for the case of base-point dependent metric $\mu_0^M$ and connection $\A_0$. On the other hand, \cite{BLM01a} assume $G$ to be abelian and also that $\mu_C^M$ satisfies two additional conditions (in the paragraph below \cite[Equ.~(2.1)]{BLM01a}), which are not necessary for Theorem~\ref{thm:MC}. Equation~\eqref{e:A_C} for the controlled connection $\A_C$ corresponds to assumption $M1$ in \cite{BLM01a}.

While the method of controlled Lagrangian or Hamiltonian systems was  developed with the aim of stabilizing a given (unstable) equilibrium, stability is not the subject of this paper and only touched upon very briefly in  Section~\ref{sec:SD-equil}. 

Further, we do not consider a second method of controlling mechanical systems, namely that of potential shaping (\cite{BCLM01}).

\subsection*{Structure of the paper}
Section~\ref{sec:SYMACT} defines the notion of a symmetry actuating force, shows that the resulting systems belong to the class of closed-loop Hamiltonian systems in the sense of \cite{CBLMW02} and provides the controlled Noether Theorem~\ref{thm:ConsLaw}. 

Section~\ref{sec:bun-pic} is mostly technical and contains many of the formulas, particularly in Propositions~\ref{prop:X^j} and \ref{prop:Xtilde}, which are subsequently used. Further, the form of the control map $j: T^*P\to\gu^*$ is found in \eqref{e:B-j-tilde}. 
This section could be shortened by assuming \eqref{e:B-j-tilde} from the beginning, but the general form of the equations is used in the discussion in Section~\ref{sec:CON-B=1}. 

Section~\ref{sec:match} contains the general matching result Theorem~\ref{thm:MC}. 

Section~\ref{sec:SD} applies Theorem~\ref{thm:MC} to the case of a semi-direct product $M\circledS G$ and a right-invariant Kaluza-Klein Hamiltonian $H_0$. The main result in this section is Theorem~\ref{thm:SD-matching}. In Section~\ref{sec:SD-equil} there are some comments on the suitability of the constructed force $F$ with regard to stabilizing a given (unstable) equilibrium.  

Section~\ref{sec:SAT} puts the satellite example in the context of Theorem~\ref{thm:SD-matching}. 

Section~\ref{sec:YM} does the same for the charged fluid example.

Section~\ref{sec:CON} contains conclusions and possible future directions. 

Section~\ref{sec:app} is an appendix and collects formulas which are used mainly in Section~\ref{sec:bun-pic}. The global formula~\eqref{app:OmK} for the canonical symplectic form on the cotangent bundle of a Riemannian manifold seems to be not very well known. (I could not find it in the literature.)

\section{Symmetry actuation}\label{sec:SYMACT}

\subsection{The mechanical system}\label{sec:SYM-MechSyst}
Let $\pi: P\to M$ be a finite dimensional principal bundle over $M$ with structure group $G$ and principal right action denoted by $r$. 


Consider the phase space $T^*P$ together with the cotangent lifted $G$-action and the canonical symplectic form $\Om^{T^*P} = -d\theta^{T^*P}$, which is $G$-invariant. The standard $G$-equivariant momentum map is denoted by $J : T^*P\to\gu^*$. 
Let $H: T^*P\to\R$ be a $G$-invariant Hamiltonian. 
It follows that $dJ.X_H = 0$.

\subsection{The control connection}\label{sec:Gamma}
There are two subbundles, namely the vertical $\ver$ and dual horizontal $\hor^*$, which are canonically associated to the principal fiber bundle $\pi: P\to M$. Explicitly , these are 
\[
 \ver=\ker T\pi\subset TP
 \quad\textup{and}\quad
 \hor^* = \textup{Ann}(\ver)\subset T^*P
\]
where the annihilator is defined as $\textup{Ann}(\ver) = \{\Pi \in T^*P: \Pi(X) = 0 \;\forall X\in\ver\}$. 
Let $\Gamma\in\Om^1(P,\gu)$ denote a connection form on $\pi: P\to M$. 
Define 
\begin{equation}
    \hor_{\Gamma} = \ker\Gamma
    \quad\textup{and}\quad
    \ver_{\Gamma}^* = \textup{Ann}(\hor_{\Gamma}).
\end{equation}
The corresponding projections are 
\begin{equation}\label{e:hpr_Gamma}
    (\hpr_{\Gamma},\vpr_{\Gamma}): TP \to \hor_{\Gamma}\oplus\ver,\quad
    (\hpr_{\Gamma}^*,\vpr_{\Gamma}^*): T^*P\to \hor^*\oplus\ver_{\Gamma}^*,
\end{equation}
and these are connection dependent isomorphisms. 
We refer to $\Gamma$ as the \emph{control connection}. 

\begin{remark} 
The purpose of $\Gamma$ is to single out the directions along which the control force $F$ acts, respectively does not act. See Definition~\ref{def:admissible}. For instance, in the satellite example, the force acts on the rotor, but not on the satellite itself. This distinction cannot be made by the mechanical connection, but necessitates another splitting of the phase space $T^*P$. 
The control connection $\Gamma$ is not to be confused with the controlled connection $\A_C$ introduced below.
\end{remark}

\subsection{Symmetry actuation and conserved quantities}
The isomorphism \eqref{e:hpr_Gamma} lifts to the tangent bundle
\begin{equation}
    T(\hpr_{\Gamma}^*,\vpr_{\Gamma}^*): 
    TT^*P\to T\hor^*\oplus T\ver_{\Gamma}^*. 
\end{equation}
Let $\tau_{TT^*P}: TT^*P\to T^*P$ be the tangent projection and note that there is an isomorphism
\[
 (\tau_{TT^*P},T\vpr_{\Gamma}^*)|\ker T\hpr_{\Gamma}^*: \ker T\hpr_{\Gamma}^*\to T^*P\times_{\textup{Ver}_{\Gamma}^*}T\ver_{\Gamma}^*
\]
of fiber bundles over $T^*P$. For $\Pi\in T^*P$ we will write this isomorphism simply as $T_{\Pi}\vpr_{\Gamma}^*: \ker T_{\Pi}\hpr_{\Gamma}^* \to T_{\textup{vpr}_{\Gamma}^*(\Pi)}\ver_{\Gamma}^*$.

\begin{definition}[Symmetry actuating force]\label{def:admissible}
We say that a vector field $F\in\X(T^*P)$ is a \emph{symmetry actuating force}, if it satisfies the following three conditions:
\begin{enumerate}
    \item 
    It is vertical with respect to the base point projection $\tau: T^*P\to P$, that is: $T\tau.F = 0$.
    \item
    It does not act along $\Gamma$-horizontal directions, that is: $T\hpr_{\Gamma}^*.F = 0$.
    \item
    It is $G$-invariant with respect to the cotangent lifted $r$-action on $T^*P$. 
\end{enumerate}
\end{definition}

\begin{proposition}\label{prop:force}
Let $\Gamma$ be a control connection. 
\begin{enumerate}
\item 
    There is a one-to-one correspondence between symmetry actuating forces $F\in\X(T^*P)$ and $G$-equivariant functions $f \in \F_G(T^*P,\gu^*)$: \\
    Given $F$: 
\begin{equation}
        f = f_F := dJ.F 
\end{equation}
    Given $f$:
\begin{equation}\label{e:f_to_F}
    F(\Pi)
    =
    F_f(\Pi)
    :=
    \Big(T_{\Pi}\vpr_{\Gamma}^*|\ker T_{\Pi}\hpr_{\Gamma}^*\Big)^{-1}
    \Big(T_{\textup{vpr}_{\Gamma}^*(\Pi)}(\tau,J)|\ver_{\Gamma}^*\Big)^{-1}
    \Big(0_{\om}; J(\Pi), f(\Pi)\Big) 
\end{equation}
where $\om = \tau(\Pi)$ and $0_{\om}$ is the $0$-section at $\om$ in $TP$.
\item
Let $f\in\F_G(T^*P,\gu^*)$
and define
$\mathcal{U}(\Pi) 
:= ((\tau,J)|\ver_{\Gamma}^*)^{-1} (\tau(\Pi), f(\Pi) ) \in \ver_{\Gamma}^*$. 
If $F = F_f$, then
\begin{equation}
    \label{e:CLU}
    F(\Pi) 
    = \textup{vl}^*(\Pi,\mathcal{U}(\Pi))
    = \dd{t}{}|_0(\Pi + t\,\mathcal{U}(\Pi))
\end{equation}
where $\textup{vl}^*: T^*P\oplus T^*P\to\ker T\tau\subset TT^*P$ is the vertical lift map.
\item
Let $F\in\X(T^*P)$ be a symmetry actuating force and consider  
\begin{equation}
\label{e:eomF}
    \dot{\Pi}_t = X_{H}(\Pi_t) + F(\Pi_t).
\end{equation}
Then
\begin{equation}\label{e:actintsym}
    \dd{t}{}J(\Pi_t) = f(\Pi_t)
\end{equation}
where $f = dJ.F$. 
\end{enumerate}
\end{proposition}

\begin{proof}
Assertion~(1).
Let $F$ be $G$-actuating.
Then
$f = f_F = dJ.F\in \F_G(T^*P,\gu^*)$ because $F$ is assumed to be $G$-invariant.

Since $J^{-1}(0) = \hor^*$, it follows that $(\tau,J)|\ver_{\Gamma}^*: \ver_{\Gamma}^*\to P\times\gu^*$ is an isomorphism. 
Now, $F\in \ker T\hpr_{\Gamma}^*\cap \ker T\tau$ implies
\[
 T_{\Pi}\vpr_{\Gamma}^*.F(\Pi)
  = 
  \Big(T_{\textup{vpr}_{\Gamma}^*\Pi}(\tau,J)|\ver_{\Gamma}^*\Big)^{-1}
  \Big(0_{\om}; J(\Pi), dJ.F(\Pi)\Big) 
\]
where $0_{\om}\in TP$ is the $0$-section at $\om = \tau(\Pi)$.  
Therefore, given $f\in \F_G(T^*P,\gu^*)$ we may define 
$F = F_f$ according to \eqref{e:f_to_F}.

Assertion~(2).
Let $\Pi\in T^*P$ and $\om = \tau(\Pi)$.
Then
\begin{align*}
    T_{\Pi}\vpr_{\Gamma}^*.F(\Pi)
    &= 
    \Big(
        T(\tau,J)|\ver_{\Gamma}^*
    \Big)^{-1}(0_{\om};J(\Pi),f(\Pi))
    \\
    &=
    \dd{t}{}|_0
    \Big(
        (\tau,J)|\ver_{\Gamma}^*
    \Big)^{-1}
    (\om, J(\Pi) + t f(\Pi)) \\
    &=
    \dd{t}{}|_0
    \Big(
    \Big(
        (\tau,J)|\ver_{\Gamma}^*
    \Big)^{-1}
    (\om, J(\Pi) )
    +
    t
    \Big(
    (\tau,J)|\ver_{\Gamma}^*
    \Big)^{-1}
    (\om,  f(\Pi))
    \Big)\\
    &= 
    \textup{vl}^*(\vpr_{\Gamma}^*(\Pi), \mathcal{U} (\Pi))
\end{align*}
where we use hat $J$ is linear in the fiber. Since 
$\textup{vl}^*(\Pi, \mathcal{U} (\Pi)) \in \ker T_{\Pi}\hpr_{\Gamma}^*$, the claim follows. 

Assertion~(3) follows from $dJ.X_{H}=0$.
\end{proof}

\begin{remark}
Equation~\eqref{e:CLU} means that \eqref{e:eomF} is a closed-loop Hamiltonian system in the sense of \cite{CBLMW02} with control subbundle $\ver_{\Gamma}^*\subset T^*P$ and control $\mathcal{U}: T^*P\to\ver_{\Gamma}^*$. 
(See the paragraph below Definition~3.1 in loc.~cit.)
\end{remark}

\begin{definition}[Symmetry actuated closed-loop Hamiltonian system]
\label{def:SymCloLoop}
Let $\pi: P\to M$, $G$, $\Gamma$ and $H$ as above.  
A \emph{symmetry actuated closed-loop Hamiltonian system} on $T^*P$ 
is defined by the equation of motion 
\begin{equation}
\label{e:Feom1}
     \dot{\Pi}_t = X_{H}(\Pi_t) + F(\Pi_t).
\end{equation}
where $F=F_f$ is a
symmetry actuating force determined by \eqref{e:f_to_F}, and where 
\begin{equation}
\label{e:defSymAct}
     f(\Pi) = -B_{\om}.dj.X_{H}(\Pi)
\end{equation}
with $\om = \pi(\Pi)$ and:
\begin{enumerate}
    \item
    $B: P\times\gu^*\to \gu^*$, 
    $(\om,p)\mapsto B_{\om}p$ is $G$-equivariant and $B_{\om}$ is an isomorphism for all $\om\in P$. That is, $\Ad(g)^*B_{\om} = B_{\om g}\circ\Ad(g)^*$. 
    \item
    $j$ is a $G$-equivariant map of the  form
    \[
     j = 
     B^{-1}\circ(\tau,\tilde{\jmath})\circ\hpr_{\Gamma}^*
     +
     (B^{-1}-\pr_2)\circ (\tau,J)
     : T^*P\to\gu^*
    \]
    for a fiber-wise linear and $G$-equivariant map $\tilde{\jmath}: \hor^*\to\gu^*$,
    and
    where $B^{-1}$ is the fiber-wise inverse of $B$. 
    That is, 
    $j(\Pi) 
    = 
    B_{\om}^{-1}.\tilde{\jmath}(\hpr_{\Gamma}^*(\Pi)) + (B_{\om}^{-1}-1).J(\Pi)$ where $\om = \tau(\Pi)$.
\end{enumerate}
\end{definition}

\begin{theorem}[Conservation law]\label{thm:ConsLaw}
If $\Pi_t$ is a solution of \eqref{e:Feom1}, 
then $(J + j)(\Pi_t)$ is constant in $t$.
\end{theorem}

\begin{proof}
Equation~\eqref{e:defSymAct} and Proposition~\ref{prop:force} combine to yield
\begin{equation}
\label{e:j-force-H_0}
    dJ.F = -B.dj.X_{H}.
\end{equation}
Further, 
\begin{align*}
    dj.F(\Pi)
    &=
    d(B^{-1}).(
        \underbrace{T\tau.F(\Pi)}_{0_{\om}},
        \underbrace{T(\tilde{\jmath}.\textup{hpr}_{\Gamma}^*).F(\Pi)}_{(\tilde{\jmath}.\textup{hpr}_{\Gamma}^*(\Pi),0)}
        )
    + d(B^{-1}-1).(\underbrace{T\tau.F(\Pi)}_{0_{\om}}, TJ.F(\Pi)) \\
    &=
    (B^{-1}-1).dJ.F(\Pi)
    =
    - (B^{-1}-1)B.dj.X_H(\Pi)
\end{align*}
where $0_{\om}\in T_{\om}P$ is the zero-section at $\om=\tau(\Pi)$ and we use fiber-wise linearity of $B$.  
Due to $dJ.X_{H} = 0$, by $G$-invariance of $H$, we obtain
\begin{align}
    \label{e:ConsLaw}
    (dJ+dj).(X_{H}+F)
    &= 
    dj.X_H
    - B.dj.X_H
    - (B^{-1}-1)B.dj.X_H
    = 0.
\end{align}
\end{proof}

\begin{remark}
Equation~\eqref{e:actintsym} means that the force acts along internal symmetry directions. At the same time, the condition $T\hpr_{\Gamma}^*.F = 0$ ensures that the force does not act in  $\hor^*$-directions, corresponding to the shape space $M$. Compare with \cite[Equ.~(1.10)]{BMS97}
or \cite[Equ.~(2.6)]{BLM01a}.
\end{remark}

\section{The bundle picture of symmetry actuated closed-loop Hamiltonian systems}
\label{sec:bun-pic}
Let $\pi: P\to M$, $G$,  $\Gamma$ and $j$ as above.
Let $\mu_0^{M}$ be a metric on $M$ and $\I_0$ a $G$-invariant locked inertia tensor on $P\times\gu$.

Fix a connection form $\A_0 \in \Om^1(P,\gu)$. 
The horizontal, vertical and respective dual spaces associated to $\A_0$ are denoted by 
\begin{equation}
    \hor_0 = \ker\A_0,\;
    \ver = \ker T\pi,\;
    \hor^{*} = \textup{Ann}(\ver),\;
    \ver_0^* = \textup{Ann}(\hor_0),
\end{equation}
where only $\hor_0$ and $\ver_0^*$ depend on the choice of connection. The corresponding projections are 
\begin{equation}
    (\hpr_0,\vpr_0): TP \to \hor_0\oplus\ver,\quad
    (\hpr_0^*,\vpr_0^*): T^*P\to \hor^*\oplus\ver_0^*,
\end{equation}
and these depend on the chosen connection $\A_0$. 

The associated Kaluza-Klein metric on $P$ is denoted by 
\[
 \mu^{P}_0 = \mu^{KK}(\mu_0^{M},\I_0,\A_0)
\]
with associated kinetic energy Hamiltonian 
\[
 H_0: 
 T^*P\to\R, 
 \Pi\mapsto \by{1}{2}\vv<\Pi, (\mu^{P}_0)^{-1}\Pi> + V\circ\tau
\] 
where $V\circ\tau: T^*P\to\R$ is a $G$-invariant potential. 
The force associated to $H_0$ via \eqref{e:defSymAct} is called $F_0 = F_f$ and  we consider, analogously to \eqref{e:Feom1}, the forced equation
\begin{equation}
    \label{e:Feom1KK}
    \dot{\Pi}_t
    = X_{H_0}(\Pi_t) + F_0(\Pi_t).
\end{equation}
Let 
\begin{equation}
    \label{e:W}
    \mathcal{W} 
    := P\times_M T^*M \oplus P\times\gu^*
\end{equation}
and consider the $(\A_0,j)$-dependent isomorphism
\begin{equation}
    \Psi_j: 
    \mathcal{W}\to T^*P,\quad
    (\om, \nu; \om, q)
    \mapsto 
    \Big((\hl^{\mathcal{A}_0})^*\Big)^{-1}(\om,\nu)
    +  
    \Big((\tau,J+j)|\ver_0^*\Big)^{-1}(\om,q)
\end{equation}
where $\hl^{\mathcal{A}_0}: P\times_M TM\to\hor$ is the  $\A_0$-horizontal lift, $\om\in P$
and
$\nu\in T^*M$ with $\pi(\om) = \tau(\nu)$. Here, and in the following, it is assumed that $j$ is chosen such that $\Psi_j$ is indeed an isomorphism. 
For $j=0$, $\Psi_j$ coincides with the $\A_0$-dependent isomorphism $\Psi_0$ defined in \eqref{app:Psi}. 

From now on we will identify $\mathcal{W} = P\times_M T^*M\oplus P\times\gu^* = P\times_M T^*M \times\gu^*$ and write elements $W\in\mathcal{W}$ as $W=(\om,\nu,q)$ where $\om\in P$, $\nu\in T^*M$, $q\in\gu^*$ and it is understood that $\pi(\om) = \tau(\nu)$. Correspondingly, we introduce the the projection maps
\[
 (\pr_1,\pr_2,\pr_3): \mathcal{W}\to P\times T^*M\times\gu^*,
 \quad
 W\mapsto (\om,\nu,q).
\]
Further, we will abbreviate $(\pr_1,\pr_2) = \pr_{1,2}$ and $(\pr_1,\pr_3) = \pr_{1,3}$. 

The isomorphism $\Psi_j$ can be decomposed as 
$\Psi_0(\om,\nu,q) 
= 
\Psi_0^{\textup{hor}}(\om,\nu) + \Psi_j^{\textup{ver}}(\om,q)$, 
where
\begin{equation}
\Psi_0^{\textup{hor}}: P\times_M T^*M\to \hor^*,\quad
 \Psi_0^{\textup{hor}}(\om,\nu)
 = \Big((\hl^{\mathcal{A}_0}_{\om})^*\Big)^{-1}\nu \in\hor^*(\om) 
\end{equation}
and 
\begin{equation}
\Psi_j^{\textup{ver}}: P\times\gu^* \to \ver_0^*,\quad
 \Psi_j^{\textup{ver}}(\om,q)
 = \Big((J + j)|\ver_0^*(\om)\Big)^{-1} q
 \in\ver_0^*(\om) 
\end{equation}
with $\ver_0^*(\om) = T_{\om}^*P\cap \ver_0^*$
and  $\hor^*(\om) = T_{\om}^*P\cap \hor^*$.

The isomorphism $\Psi_0$ induces a $G$-action on $\mathcal{W}$. This action leaves the induced symplectic form $\Psi_0^*\Om^{T^*P}$ invariant and has a momentum map 
\[
 J_{\mathcal{W}} := \pr_3 = \Psi_0^*J: (\om,\nu,q)\mapsto q.  
\]

The isomorphism $\Psi_j$ induces a $G$-action on $\mathcal{W}$. This action leaves the induced symplectic form $\Psi_j^*\Om^{T^*P}$ invariant and has a momentum map 
\begin{align}
    \notag
    \Psi_j^*J
    &= (\Psi_0^{\textup{hor}}\circ\pr_{1,2}+\Psi_j^{\textup{ver}}\circ\pr_{1,3})^*J
    = (\Psi_j^{\textup{ver}}\circ\pr_{1,3})^*J\\
    &= J_{\mathcal{W}} - j\circ(\tau,J+j)|_{\textup{Ver}_0^*}^{-1}\circ\pr_{1,3}
    =  J_{\mathcal{W}} - (\Psi_j^{\textup{ver}}\circ\pr_{1,3})^*j
    \label{e:J^j}
\end{align}
Since $\Psi_j$ is assumed to be an isomorphism, the $\om$-dependent map $(J\circ\Psi_j^{\textup{ver}})_{\om}: \gu^*\to\gu^*$ is an isomorphism as-well. Therefore, \eqref{e:J^j} implies that the same is true for $(1-j\circ\Psi_j^{\textup{ver}})_{\om}$, for each $\om\in P$.  Further, the map $1-j\circ\Psi_j^{\textup{ver}}: P\times\gu^*\to\gu^*$ is $G$-equivariant. 

The actions induced by $\Psi_0$ and $\Psi_j$ coincide, since $j$ is equivariant. Therefore, we will simply speak of the $G$-action on $\mathcal{W}$. It is given by
\[
 (\om,\nu,q).g = (\om.g,\nu,\Ad(g)^*.q)
\]
for $g\in G$. 

\begin{remark}
The momentum maps $J_{\mathcal{W}}$ and $\Psi_j^*J$, associated respectively to $\Psi_0^*\Om^{T^*P}$ and $\Psi_j^*\Om^{T^*P}$, do not necessarily coincide for $j\neq 0$. However, with assumption \eqref{e:B-j-tilde} below it follows, that $\Psi_0 = \Psi_j$.  
\end{remark}

The infinitesimal generator associated to an element $Y\in\gu$ shall be denoted by $\zeta^{\mathcal{W}}_Y$ where $\zeta^{\mathcal{W}}: \gu\to\X(\mathcal{W})$ is the fundamental vector field mapping. 

Consider the mapping  $(\id,D^0): P\times\gu^*\to P\times_M T^*M$, $(\om,p)\mapsto(\om,D_{\om}^0.p)$, where $D_{\om}^0: \gu^*\to T^*M$ is defined as
\begin{equation}
    \label{e:D^0}
    D_{\om}^0 
    := 
    (\pr_2\circ(\hlA)^*\circ\hpr_{\Gamma}^*\circ\Psi_0^{\textup{ver}})_{\om}
    : \gu^*\to T_x^*M
\end{equation}
where $x=\pi(\om)$.

\begin{proposition}\label{prop:X^j}
Assume $\Psi_j$ is an isomorphism. 
Consider \eqref{e:Feom1KK}.
Let
$H^j := \Psi_j^*H_0$,
$X^j := \Psi_j^*X_{H_0}$ and $F^j := \Psi_j^*F_0$. Use the projections $(\pr_1,\pr_2,\pr_3)$ to set 
\[
 (T\pr_1,T\pr_2,d\pr_3).X^j = (X^j_1,X^j_2,X^j_3)\in TP\times TT^*M\times \gu^*
\]
and 
\[
 (T\pr_1,T\pr_2,d\pr_3).F^j = (F^j_1,F^j_2,F^j_3)\in TP\times TT^*M\times \gu^*.
\]
Let $\hl_{\mu_0^M}^*: T^*M\oplus TM\to \hor(\mu_0^M) \subset TT^*M$ denote the Riemannian horizontal lift and $\textup{vl}^*: T^*M\oplus T^*M \to \ver(\tau_{T^*M}) = \ker T\tau_{T^*M}\subset TT^*M$ the vertical lift. Then
\begin{align}
    X^j_1(W)
    &= 
    \hl_{\om}^{\mathcal{A}_0}((\mu_0^M)^{-1}\nu) 
    + \zeta_{\tilde{\mathbb{I}}^{-1}p}(\om)
    \\
    \label{e:X_2^j}
    X^j_2(W)
    &=
    (\hl_{\mu_0^M}^*)_{\nu}((\mu_0^M)^{-1}\nu)
    - \textup{vl}_{\nu}^*(d^{\textup{hor}}H^j(W) + L(W)) \\
    \label{e:X_3^j}
    X^j_3
    &= 
    (1 - j\circ\Psi_j^{\textup{ver}})^{-1}.dj.T\Psi_j^{\textup{ver}}.(X^j_1,0)\\
    F^j_1 
    &= 0\\
    \label{e:F_2^j}
    F^j_2(W)
    &= 
    \vl_{\nu}^*\Big(D_{\om}^0.B_{\om}.dj.T\Psi_j.X^j(W)\Big)
    \\
    F^j_3(W)
    &= 
    -(1 + j\circ\hpr_{\Gamma}^*\circ(\Psi_0^{\textup{ver}})_{\om}\circ B_{\om}).dj.T\Psi_j.X^j(W) 
\end{align}
for $W = (\om,\nu,p)\in\mathcal{W} = P\times_{M}T^*M\times\gu^*$ and 
where $\tilde{\mathbb{I}}: \gu\to\gu^*$ is the $\om$-dependent inertia tensor
\begin{equation}
    \label{e:IC}
    \tilde{\mathbb{I}}
    := (1-j\circ\Psi_j^{\textup{ver}})^{-1}\mathbb{I}_0. 
\end{equation}
The operators $d^{\textup{hor}}$ and $L$ are defined below in \eqref{e:d^hor} and \eqref{e:L}. 
\end{proposition}

\begin{proof}
Let $\theta^{T^*P}$ be the canonical one-form on $T^*P$ with $\Om^{T^*P} = -d\theta^{T^*P}$. Let $W=(\om,\nu,p)\in\mathcal{{W}}$ and $X = (X_1,X_2,X_3)\in T_{\om}P\times T_{\nu}T^*M\times\gu^*$ with $T\pi.X_1 = T\tau.X_2 = u$. Further, let $X_1 = \hl_{\om}^{\mathcal{A}_0}(u) + \zeta_Y(\om)$. Using \eqref{e:J^j},
\begin{align*}
    \Big(\Psi_j^*\theta^{T^*P}\Big)_W(X_1,X_2,X_3)
    &=
    \vv<\nu,u> 
        + \vv<\Psi_j^{\textup{ver}}(\om,q),\zeta_Y(q)>\\
    &=
    \theta^{T^*M}_{\nu}(X_2) 
        + \vv<(J_{\mathcal{W}}-(\Psi_j^{\textup{ver}}\circ\pr_{1,3})^*j)(W),Y>
\end{align*}
where $\theta^{T^*M}$ is the canonical one-form on $T^*M$ with $\Om^{T^*M} = -d\theta^{T^*M}$.
Since $\Om^j := \Psi_j^*\Om^{T^*P} = -d\Psi_j^*\theta^{T^*P}$, and using the Maurer-Cartan formula
\[
 \CurvA = d\A_0 + \by{1}{2}[\A_0,\A_0]_{\wedge}
\]
as-well as $\pr_1^*\A_0(X) = Y$, 
it follows that 
\begin{align*}
    \Om^j
    = 
    \pr_2^*\Om^{T^*M}
    &
    - 
    \Big\langle 
        d\Big(J_{\mathcal{W}}-(\Psi_j^{\textup{ver}}\circ\pr_{1,3})^*j\Big)
        \,\overset{\wedge}{,}\,
        \pr_1^*\A_0 
    \Big\rangle \\
    &
    -
    \Big\langle
        J_{\mathcal{W}}-(\Psi_j^{\textup{ver}}\circ\pr_{1,3})^*j 
            ,
        \pr_1^*(\CurvA - \by{1}{2}[\A_0,\A_0]_{\wedge})
    \Big\rangle.
\end{align*}
This formula will be used throughout the rest of the proof without further reference.

Equation~\eqref{e:J^j} implies that 
\begin{align*}
    \vv<\mathbb{I}_0^{-1} J \Psi_j^{\textup{ver}}. p, J \Psi_j^{\textup{ver}}. p>
    &= 
    \vv<\mathbb{I}_0^{-1}(1 - j\Psi_j^{\textup{ver}}). p, (1-j\Psi_j^{\textup{ver}}). p>
    =
    \vv<\mathbb{I}_j^{-1} p, p>
\end{align*}
where the $\om$-dependence is understood implicitly and 
\[
 \mathbb{I}_j^{-1}
 := (1-j\Psi_j^{\textup{ver}})^*\mathbb{I}_0^{-1}(1 - j\Psi_j^{\textup{ver}}).
\]
This yields 
$H^j(W) 
= \by{1}{2}\vv<\nu,(\mu_0^M)^{-1}\nu> 
+
\by{1}{2}\vv<\mathbb{I}_0^{-1} J \Psi_j^{\textup{ver}}. p, J \Psi_j^{\textup{ver}}. p>
= \by{1}{2}\vv<\nu,(\mu_0^M)^{-1}\nu> + \by{1}{2}\vv<p,\mathbb{I}_j^{-1}p>$.

Since $X^j(W)\in T_W\mathcal{W}$ it follows that $T\pi.X^j_1(W) = T\pi.T\pr_1.X^j(W) = T\tau_{T^*M}.T\pr_2.X^j(W) = T\tau_{T^*M}.X^j_2(W) =: u \in T_xM$ where $x = \pi(\om)$.

To calculate $X_1^j$:  
Decompose $X_1^j(W)$ as $X_1^j(W) = \hlA(\om,u) + \zeta_Y(\om)$ with $Y = \A_0.X_1^j(W)$.
Consider $\eta_2\in T_x^*M$ and 
\[
 Z := \Big(0 , \vl^*_{\nu}(\eta_2), 0 \Big) \in T_W\mathcal{W}.
\]
Then, using the $\mu_0^M$-connector $K_0^*: TT^*M\to T^*M$
together with \eqref{app:OmK},
\begin{equation}
    dH^j.Z
    =
    \vv<(\mu_0^M)^{-1}\nu, \eta_2 >
    =
    \Om^j(X^j(W),Z)
    =
    \Om^{T^*M}(X_2^j(W),\vl_{\nu}^*(\eta_2))
    =
    \vv<u, \eta_2>
\end{equation}
whence $u = (\mu_0^M)^{-1}\nu$. Consider now $\dot{p}_2\in \gu^*$ and 
\[
 Z = \Big(0, 0, \dot{p}_2 \Big) \in T_W\mathcal{W}.
\]
Then
\begin{align*}
    dH^j(W).Z
    &=
    \vv<\mathbb{I}_j^{-1}p, \dot{p}_2 >
    =
    \Om^j(X^j(W),Z)
    =
    \vv<d\Big(J_{\mathcal{W}}-(\Psi_j^{\textup{ver}}\circ\pr_{1,3})^*j\Big).Z, \A_0.X_1^j(W)> \\
    &=
    \vv< \dot{p}_2 - dj.T\Psi_j^{\textup{ver}}.(0,\dot{p}_2), Y>
    =
    \vv< (1 - j\Psi_j^{\textup{ver}})_{\om}.\dot{p}_2, Y>    
    =
    \vv< \dot{p}_2, (1 - j\Psi_j^{\textup{ver}})_{\om}^*.Y>  
\end{align*}
where the penultimate equality follows because $j$ is fiberwise linear, and we have indicated the $\om$-dependence. 
Therefore,
\[
 Y 
 = \mathbb{I}_0^{-1}(1 - j\Psi_j^{\textup{ver}}).p
 = \tilde{\mathbb{I}}^{-1}p.
\]

To calculate $X_2^j$:  
Decompose $X_2^j(W)$ as $X_2^j(W) = (\hl_{\mu_0^M}^*)_\nu(u) + \vl^*_{\nu}(\eta)$ with $\eta\in T_x^*M$. Consider now $u_2\in T_xM$ and 
\[
 Z = \Big(\hlA(\om,u_2), (\hl_{\mu_0^M}^*)_{\nu}(u_2), 0\Big) \in T_W\mathcal{W}.
\]
Then
\begin{equation}
    \label{e:d^hor}
    \Om^j(X^j,Z)
    = dH^j(W).\Big(\hlA(\om,u_2), (\hl_{\mu_0^M}^*)_{\nu}(u_2), 0\Big)
    =: d^{\textup{hor}}H^j(W).u_2
\end{equation}
and, using the $\mu_0^M$-connector $K_0^*: TT^*M\to T^*M$
together with \eqref{app:OmK},
\begin{align*}
    \Om^j(X^j,Z)
    &= 
    \Om^{T^*M}(X_2(W), Z)
    - 
    \langle 
        d J_{\mathcal{W}}
            \,\overset{\wedge}{,}\,
        \pr_1^*\A_0
    \rangle (X_2(W), Z)\\
    &\phantom{==}
    +
    \langle
        (\Psi_j^{\textup{ver}}\circ\pr_{1,3})^*dj
            \,\overset{\wedge}{,}\,
        \pr_1^*\A_0 
    \rangle (X_2(W), Z)\\
    &\phantom{==}
    -
    \langle
        (J_{\mathcal{W}}-(\Psi_j^{\textup{ver}}\circ\pr_{1,3})^*j)(W) 
            ,
        \CurvA(u,u_2) 
    \rangle\\
    &= 
    -
    \vv<(K_0^*)_{\nu}(X_2(W)), u_2> 
    - 
    \vv< dj.T\psi_j^{\textup{ver}}.T\pr_{1,3}.Z, Y>\\
    &\phantom{==}
    -
    \vv< p - j(\Psi_j^{\textup{ver}}(\om,p)), \CurvA(u,u_2) >\\
    &= 
    -
    \vv<\eta , u_2> 
    - 
    \vv< dj.T\psi_j^{\textup{ver}}.(\hlA(\om,u_2),0), \tilde{\mathbb{I}}^{-1}p>\\
    &\phantom{==}
    -
    \vv< (1 - j\circ\Psi_j^{\textup{ver}}(\om)).p, \CurvA(u,u_2) >
\end{align*}
Define $L(W)\in T_x^*M$ through 
\begin{equation}
    \label{e:L}
    \vv<L(\om,\nu,p),u_2>
    = 
    \vv< dj.T\psi_j^{\textup{ver}}.(\hlA(\om,u_2),0), \tilde{\mathbb{I}}^{-1}p>
    +
    \vv< (1 - j\circ\Psi_j^{\textup{ver}}(\om)).p, \CurvA(u,u_2) >.
\end{equation}
Hence,  $\eta = -d^{\textup{hor}}H^j(W) - L(W)$. 

To calculate $X_3^j$:  Consider now $Y_2\in \gu$ and 
\[
 Z 
 = \zeta^{\mathcal{W}}_{Y_2}(\om,\nu,p)
 = \Big(\zeta^P_{Y_2}(\om), 0, \ad(Y_2)^*.p \Big) 
 \in T_W\mathcal{W}.
\]
Then, because $H^j$ is $G$-invariant
and 
$dJ.T\Psi_j. \zeta^{\mathcal{W}}_{Y_2}(W)) 
= dJ.\zeta^{T^*P}_{Y_2}(\Psi_j(W)))
= \ad(Y_2)^*.J(\Psi_j(W)))$, 
\begin{align*}
    0
    &= 
    dH^j. \zeta^{\mathcal{W}}_{Y_2}(W)
    = 
    \Om^j(X^j(W),  \zeta^{\mathcal{W}}_{Y_2}(W))\\
    &=
    -
    \Big\langle 
        d\Big(J_{\mathcal{W}}-(\Psi_j^{\textup{ver}}\circ\pr_{1,3})^*j\Big)
        \,\overset{\wedge}{,}\,
        \pr_1^*\A_0 
    \Big\rangle . \Big(X^j(W),  \zeta^{\mathcal{W}}_{Y_2}(W)\Big)
    +
    \Big\langle 
        p - j(\Psi_j^{\textup{ver}}(\om,p))
        ,
        [Y,Y_2]
    \Big\rangle\\
    &=
    -
    \Big\langle 
        d\Big(J_{\mathcal{W}}-(\Psi_j^{\textup{ver}}\circ\pr_{1,3})^*j\Big).X^j(W)
        ,
        Y_2
    \Big\rangle
    +
    \Big\langle 
        dJ.T\Psi_j. \zeta^{\mathcal{W}}_{Y_2}(W))
        ,
        Y
    \Big\rangle 
    +    
    \Big\langle 
        J(\Psi_j(W))
        ,
        [Y,Y_2]
    \Big\rangle\\
    &=
    -
    \Big\langle 
        dJ_{\mathcal{W}}.X^j(W)
        - dj.T\Psi_j^{\textup{ver}}.(X^j_1(W),X^j_3(W))
        ,
        Y_2
    \Big\rangle
\end{align*}
whence 
\begin{align*}
 X_3^j(W)
 &= dJ_{\mathcal{W}}.X^j(W)
 = dj.T\Psi_j^{\textup{ver}}.(X^j_1(W),X^j_3(W))\\
 &= dj.T\Psi_j^{\textup{ver}}.(X^j_1(W),0)
  + dj.T\Psi_j^{\textup{ver}}.(0,X^j_3(W)).
\end{align*}
Now,
since $j\circ\Psi_j^{\textup{ver}}: P\times\gu^*\to\ver_0^*$ is linear in $\gu^*$, it follows that
$dj.T\Psi_j^{\textup{ver}}.(0,X_3^j(W)) 
= \dd{t}{}|_0(j\circ\Psi_j^{\textup{ver}})(\om,t X_3^j(W))
= (j\circ\Psi_j^{\textup{ver}})_{\om}.X_3^j(W)$, whence
\[
 (1 - j\circ\Psi_j^{\textup{ver}}).X_3^j
 = dj.T\Psi_j^{\textup{ver}}.(X^j_1(W),0)
\]
and the $\om$-dependent map $1 - j\circ\Psi_j^{\textup{ver}}: \gu^*\to\gu^*$ is  invertible.

To calculate $F_1^j$:  
Note that $\pr_1(W) = \om = (\tau\circ\Psi_j)(W)$. Therefore,
\[
 F_1^j
 = T\pr_1.F^j
 = T\pr_1.T\Psi_j^{-1}.(F\circ\Psi_j)
 = T\tau.(F\circ\Psi_j)
 = 0
\]
since $F_0$ is vertical. 

To calculate $F_2^j$:  Let $\mathcal{U}^j := \Psi_j^{-1}\circ\mathcal{U}\circ\Psi_j$ where $\mathcal{U}: T^*P\to\ver_{\Gamma}^*$ is defined in Proposition~\ref{prop:force}. Equation~\ref{e:CLU} then yields 
\begin{align*}
 F_2^j(W)
 &= \dd{t}{}|_0(\nu + t (\pr_2\circ\mathcal{U}^j)(W)) \\
 &= \vl^*_{\nu}
    \Big(
    (\pr_2\circ\Psi_j^{-1}\circ(\tau,J)|_{\textup{Ver}_{\Gamma}^*}^{-1})(\om,f(\Psi_j(W))
    \Big) \\
 &= \vl^*_{\nu}
    \Big(
    (\pr_2\circ(\hlA)^*\circ\hpr_0^*\circ(J|\ver_{\Gamma}^*(\om))^{-1})(-B_{\om}.dj.(X_{H_0}(\Psi_j(W))))
    \Big). 
\end{align*}
Since 
$\hpr_{\Gamma}^*\circ\Psi_0^{\textup{ver}} 
= -\hpr_0^*\circ\vpr_{\Gamma}^*\circ\Psi_0^{\textup{ver}} 
= -\hpr_0^*\circ(J|\ver_{\Gamma}^*)^{-1}$,
it follows that
\[
 F_2^j(W)
 = \vl^*_{\nu}
    \Big(
    (\pr_2\circ(\hlA)^*\circ\hpr_{\Gamma}^*\circ\Psi_0^{\textup{ver}})_{\om}
    (B_{\om}.dj.T\Psi_j.(X^j(W)))
    \Big). 
\]

To calculate $F_3^j$:  
Equations~\eqref{e:j-force-H_0} and \eqref{e:J^j} imply that 
\begin{align*}
    F_3^j
    &=
    dJ_{\mathcal{W}}.\Psi_j^*F
    =
    d\Big((\Psi_j^{-1})^*J_{\mathcal{W}}\Big).(F\circ\Psi_j)\\
    &=
    d\Big(J + (\Psi_j^{-1})^*(\Psi_j^{\textup{ver}}\circ\pr_{1,3})^*j\Big)
        .(F\circ\Psi_j)\\
    &=
    d\Big(J + \vpr_0^*j\Big)
        .(F\circ\Psi_j)\\
    &=
    dJ.(F\circ\Psi_j)
    + dj.(F\circ\Psi_j)
    - dj.T\hpr_0^*.(F\circ\Psi_j)\\
    &=
    -dj.(X_{H_0}\circ\Psi_j)
    -d(j\circ\hpr_0^*)(F\circ\Psi_j)\\
    &=
    -dj.T\Psi_j.X^j
    -\dd{t}{}|_0(j\circ\hpr_0^*)(\Psi_j + t\mathcal{U}\circ\Psi_j)\\
    &= 
    -dj.T\Psi_j.X^j
    -(j\circ\hpr_0^*\circ J|_{\textup{Ver}_{\Gamma}^*(\om)}^{-1})(f\circ\Psi_j)\\
    &= 
    -\Big(
     1+
     j\circ\hpr_{\Gamma}^*\circ\Psi_0^{\textup{ver}}\circ B_{\om}
     \Big)
     dj.T\Psi_j.X^j
\end{align*}
where we use again that
$\hpr_{\Gamma}^*\circ\Psi_0^{\textup{ver}} 
= -\hpr_0^*\circ\vpr_{\Gamma}^*\circ\Psi_0^{\textup{ver}} 
= -\hpr_0^*\circ(J|\ver_{\Gamma}^*)^{-1}$.
\end{proof}

In order to make use of the conservation law \eqref{e:ConsLaw}, we define the map 
\begin{equation}
    \label{e:Phi^j}
    \Phi_j:
    \W\to\W, \quad
    (\om,\nu,q)
    \mapsto
    (\om,\nu, q - (j\circ\Psi_0^{\textup{hor}})(\om,\nu) ).
\end{equation}

\begin{proposition}\label{prop:Xtilde}
Consider \eqref{e:Feom1KK}.
Let
$X^{jj} := \Phi_j^*\Psi_j^*X_{H_0}$ and $F^{jj} := \Phi_j^*\Psi_j^* F_0$
Then
\begin{align}
    \label{e:tilde1}
    T\pr_1.\Big(X^{jj} + F^{jj}\Big) 
    &= X^j_1\circ\Phi_j 
    \\
    \label{e:ConsLawTilde}
    dJ_{\mathcal{W}}.\Big(X^{jj} + F^{jj}\Big) 
    &= d\pr_3.\Big(X^{jj} + F^{jj}\Big)
    = 0.
\end{align}
Let $B$ be defined by
\begin{equation}
    \label{e:B-j-tilde}
    B_{\om} 
    := 
    \Big(
        1 + \tilde{\jmath}\circ\hpr_{\Gamma}^*\circ\Psi_0^{\textup{ver}}
    \Big)_{\om}. 
\end{equation}
It follows that 
\begin{equation}
\label{e:j-ver-0}
    j|\ver_0^* = 0, 
\end{equation}
$\Psi_j = \Psi_0$ and $j\circ\Psi_0^{\textup{ver}} = 0$,
and
\begin{equation}
    \label{e:tilde2}
    T\pr_2.\Big(X^{jj} + F^{jj}\Big) 
    = 
    X^j_2\circ\Phi_j
    +
    \vl_{\nu}^*
    \Big(
        D_{\om}^0.B_{\om}.dj.T\Psi_0^{\textup{hor}}.(X_1^j\circ\Phi^j,X_2^j\circ\Phi^j)
    \Big).
\end{equation}
\end{proposition}

\begin{proof}
Equation~\ref{e:tilde1} is immediate from Proposition~\ref{prop:X^j}.

Equation~\eqref{e:ConsLawTilde}:
Equation~\eqref{e:J^j} yields
\[
    \Psi_j^*(J+j)
    =
    J_{\mathcal{W}} - (\Psi_j^{\textup{ver}}\circ\pr_{1,3})^*j + \Psi_j^*j
    =
    J_{\mathcal{W}} + (\Psi_0^{\textup{hor}}\circ\pr_{1,2})^*j.
\]
On the other hand, because $J_{\mathcal{W}}(\om,\nu,p) = p$, 
\[
 (\Phi_j^{-1})^*J_{\mathcal{W}}(\om,\nu,p)
 = J_{\mathcal{W}}(\om,\nu,p) 
    + (j\circ \Psi_0^{\textup{hor}}\circ\pr_{1,2})(\om,\nu,p).
\] 
Hence 
$\Phi_j^*\Psi_j^*(J+j) = J_{\mathcal{W}}$, and 
\eqref{e:ConsLaw} implies (independently of assumption~\eqref{e:B-j-tilde})
\begin{align*}
    dJ_{\mathcal{W}}.(X^{jj}+F^{jj})
    &=
    d\Phi_j^*\Psi_j^*(J+j).\Phi_j^*\Psi_j^*(X_{H_0}+F)
    =
    d(J+j).(X_{H_0}\circ\Psi_j\circ\Phi_j + F\circ\Psi_j\circ\Phi_j)\\
    &=
    0.
\end{align*}
Alternatively, this result can be shown by a direct, but lengthy, calculation using the formulas in Proposition~\ref{prop:X^j}.

Equation~\eqref{e:j-ver-0} follows from the definition of $j$. 

Equation~\ref{e:tilde2}:
Proposition~\ref{prop:X^j} implies 
$T\pr_2.F^{jj}
= T\pr_2.T\Phi_j^{-1}.(F^j\circ\Phi_j)
= F_2^j\circ\Phi_j
= \vl_{\nu}^*(D^0.B.d(j\circ\Psi_j).(X^j\circ\Phi_j))
=
\vl_{\nu}^*
    (
        D_{\om}^0.B_{\om}.dj.T\Psi_0^{\textup{hor}}.(X_1^j\circ\Phi^j,X_2^j\circ\Phi^j)
    )$.
\end{proof}

\begin{remark}
Equation~\eqref{e:j-ver-0} is the reason for introducing $B$ in Definition~\ref{def:SymCloLoop}. Due to \eqref{e:B-j-tilde}  $B$ is determined by $\tilde{\jmath}: \hor^*\to\gu^*$.
\end{remark}

\section{Matching}\label{sec:match}
Let the notation be as in Section~\ref{sec:bun-pic}. 
Consider equation~\eqref{e:Feom1KK} and assume $B$ is given by \eqref{e:B-j-tilde}.
Because of \eqref{e:j-ver-0} this implies that $\Psi_j = \Psi_0$. 


In the following, we are looking for a $G$-equivariant bundle isomorphism $\Phi: T^*P\to T^*P$, which maps \eqref{e:Feom1KK} to a Hamiltonian system (at least, when restricted to the pre-image of $J^{-1}(0)$ under this isomorphism).
The target Hamiltonian system should be of Kaluza-Klein form with respect to controlled data $(\mu_C^M,\I_C,\A_C)$, which are to be specified, and the unchanged potential $V$. 
Equivalently, the goal is to explicitly find a $G$-equivariant isomorphism $\Phi_C$
\begin{equation}
\xymatrix{
 T^*P 
 &\W
 \ar@{->}[l]_{\Psi_0} 
 &\W
 \ar@{->}[l]_{\Phi_j} 
 \ar@{..>}[r]^{\Phi_C} 
   &\W
 \ar@{->}[r]^{\Psi_C} 
 & T^*P
}    
\end{equation}
and a force $F_C$
such that \eqref{e:thm1}, \eqref{e:thm2} and \eqref{e:thm3} hold, and where $\Psi_C$ is the connection dependent isomorphism associated to the controlled connection $\A_C$.

\subsection{Matching conditions}
As above, $H_0$ denotes the Kaluza-Klein Hamiltonian associated to $(\mu_0^M,\I_0,\A_0)$ and the  potential $V$. 

Assume $H_C$ is a Kaluza-Klein Hamiltonian associated to $(\mu_C^M,\I_C,\A_C)$ and the  potential $V$. Let $\Psi_C: \W\to T^*P$ be the isomorphism corresponding to $\A_C$.
Let $\Phi_C = (\var_C,S)$, where $\var_C: P\times_M T^*M\to P\times_M T^*M$ and $S:\gu^*\to\gu^*$ are $G$-equivariant isomorphisms, and $\var_C$ is fiber-preserving. 
Let $F_C\in\X(T^*P)$ with
\begin{align}
    \label{e:thm1}
    T\tau.F_C &= 0\\
    \label{e:thm2}
    F_C|J^{-1}(0) &= 0
\end{align}
and define $F^C := \Psi_C^*F_C$.

\begin{theorem}\label{thm:MC}
Let $F_0$ be as in equation~\eqref{e:Feom1KK}.
The following are equivalent: 
\begin{enumerate}
    \item 
$\Phi^*(X_{H_C}+F_C) = X_{H_0}+F_0$
for 
$\Phi 
:= \Psi_C\circ(\var_C,S)\circ\Phi_j^{-1}\circ\Psi_0^{-1}$.
\item
The matching conditions
\begin{align}
    \label{e:mu_C^M}
    \mu_C^M &:= \var_C.\mu_0^M: TM\to T^*M\\
    \label{e:I_C}
    \I_C
    &:= S.\I_0\\
    \label{e:A_C}
    \A_C
    &:= \A_0 + \I_0^{-1}.B^{-1}.\tilde{\jmath}.\mu_0^P.\hpr_0\\
    \label{e:thm_MC_2}
    (X_2^C+F_2^C)\circ\Phi_C 
    &= T(\pr_2\circ\var_C).((X_1^j,X_2^{j}+F_2^{j})\circ\Phi_j)
\end{align}
hold.
\end{enumerate}
\end{theorem}

\begin{proof}
Using Proposition~\ref{prop:Xtilde}, 
it follows that item (1) is equivalent to 
$T\Phi_C.(X^{jj}+F^{jj}) 
= (T\var_C.(X_1^{jj},X_2^{jj}+F_2^{jj}), TS.0)
= (X^C+F^C)\circ\Phi_C = (X_1^C,X_2^C+F_2^C,0)\circ\Phi_C$,
which is hence equivalent to
$(X_1^C,X_2^C+F_2^C)\circ\Phi_C 
= T\var_C.((X_1^j,X_2^{j}+F_2^{j})\circ\Phi_j)$. 
This yields the equations 
\begin{equation}
    \label{e:proof_MC_1}
    X_1^C\circ(\var_C,S) = X_1^j\circ\Phi_j
\end{equation}
and \eqref{e:thm_MC_2}. 
Now, the former can be given explicitly as
\begin{align}
    \label{e:proofMC1}
     (\hl^{\mathcal{A}_C}_{\om})((\mu_C^M)^{-1}\mu)
    + \zeta_{\mathbb{I}_C^{-1}Sq}(\om)
    &=
    X_1^C(\om,\mu, S q)
    =
    X_1^j(\om,\nu,p)\\
    &=
    (\hl^{\mathcal{A}_0}_{\om})((\mu_0^M)^{-1}\nu)
    + \zeta_{\mathbb{I}_0^{-1}p}(\om)
        \notag
\end{align}
where $(\om,\mu, S q) = \Phi_C(\om,\nu,q)$ and $(\om,\nu,p) = \Phi_j(\om,\nu,q) = (\om,\nu,q-j\Psi_0^{\textup{hor}}(\om,\nu))$.

Equation~\eqref{e:proofMC1} necessitates $(\mu_0^M)^{-1}\nu = (\mu_C^M)^{-1}\mu = ((\mu_C^M)^{-1}\var_C)\nu$, since the fundamental vector fields are annihilated by the projection $T\pi: TP \to TM$, whence  \eqref{e:mu_C^M} follows.
Further, setting $\nu=0$, equality holds in \eqref{e:proofMC1} if, and only if, \eqref{e:I_C} holds. Setting $q=0$ and applying $\A_C$ to both sides of \eqref{e:proofMC1}, equality holds if, and only if,  \eqref{e:A_C} holds.
\end{proof}

\begin{remark}
Assume that $M$ is parallelizable such that $TM = M\times\mo$ and that 
\begin{equation}
    \label{e:inv-cond}
    T( D^0\circ (\id,B)).(X_1^j, 0)
    = 0.
\end{equation}
Let  $S: \gu^*\to\gu^*$ be an $\Ad(g)^*$-equivariant isomorphism and consider
\begin{align}
\label{e:MCvarC}
    &\var_C
    := 
    \id_P\times_M
    (\id_{T^*M} + D^0\circ(\pr_1,\tilde{\jmath}.\Psi_0^{\textup{hor}}))^{-1}:
    P\times_M T^*M  = P\times\mo\times\gu^* 
    \to 
    P\times\mo\times\gu^*.
\end{align}
Set $\Phi_C = (\var_C,S)$. 
Let $\mu_C^M$, $\I_C$ and $\A_C$ be defined by \eqref{e:mu_C^M}, \eqref{e:I_C} and \eqref{e:A_C}.
For
$(\om,\mu,\tilde{q}) = \Phi_C(\om,\nu,q)$ and $u = (\mu_0^M)^{-1}(\nu)$,
assume the matching conditions 
\begin{equation}
    \label{e:MCthm1}
    \Big(\hl_{\mu_0^M}^*\Big)_\nu\Big(u\Big)
    +\vl_{\nu}^*\Big<
        j.\Psi_0^{\textup{hor}}(\om,\nu),i(u)\CurvA
        \Big>
    =
    \Big(
    \hl_{\mu_C^M}^*
    \Big)_{\mu}\Big(u\Big)
\end{equation}
and
\begin{equation}
    \label{e:MCthm2}
    d^{\textup{hor}}H^j(\om,\nu,q-j.\Psi_0^{\textup{hor}}(\om,\nu))
    =
    d^{\textup{hor}}H^C(\om,\mu,\tilde{q})
\end{equation}
hold, where $H^C = \Psi_C^*H_C$.
The corresponding force is then given by 
\begin{equation}
    \label{e:F_2^C}
    F_2^C(\om,\mu,\tilde{q})
    :=
    \vl_\mu^*\Big(
    \vv<
        \tilde{q},
        i(\tilde{u})\Curv^{\mathcal{A}_C}   
        - (S^{-1})^*i(\tilde{u})\CurvA>
     \Big)
\end{equation}
where $(\om,\mu,\tilde{q})\in\W$ and $\tilde{u} = (\mu_C^M)^{-1}(\mu)$
and $H_C$ is the Kaluza-Klein Hamiltonian associated to $(\mu_C^M,\I_C,\A_C)$ and the (unchanged) potential function $V$. 
It follows that $H_C$, $F_C := (\Psi_C^{-1})^*(0,F^C_2,0)$ and $\Phi := \Psi_C\circ(\var_C,S)\circ\Phi_j^{-1}\circ\Psi_0^{-1}$ satisfy 
\begin{align}
    & T\tau.F_C = 0\\
    & F_C|J^{-1}(0) = 0\\
    \label{e:thm3}
    & \Phi^*(X_{H_C}+F_C) = X_{H_0}+F_0
\end{align}
Indeed, 
fix $(\om,\nu,q)\in\W$ and let $(\om,\nu,p) = (\om,\nu,q-j.\Psi_0^{\textup{hor}}(\om,\nu)) = \Phi_j(\om,\nu,q)$. 
Because of assumption \eqref{e:inv-cond}, we have
\begin{align*}
    F_2^j(\om,\nu,p)
    &=
    \vl_{\nu^*}\Big(
    D_{\om}^0.B_{\om}.d(j\circ\Psi_0^{\textup{hor}}).(X_1^j,X_2^j)(\om,\nu,p)
    \Big)\\
    &=
    T\Big(D^0\circ(\id,B)\Big)
    .\Big(0, T(j\circ\Psi_0^{\textup{hor}}).(X_1^j,X_2^j)\Big)
    (\om,\nu,p)
    \\
    &=
    T\Big(D^0\circ(\id,B)\Big)
    .\Big(X_1^j, T(j\circ\Psi_0^{\textup{hor}}).(X_1^j,X_2^j)\Big)
    (\om,\nu,p)\\
    &= 
    T\Big(
    D^0\circ(\id,B)\circ(\pr_1,j\circ\Psi_0^{\textup{hor}})
    \Big).\Big(X_1^j,X_2^j\Big)
    (\om,\nu,p)\\
    &= 
    T\Big(
        D^0\circ(\pr_1,\tilde{\jmath}\circ\Psi_0^{\textup{hor}})
    \Big).\Big(X_1^j,X_2^j\Big)
    (\om,\nu,p)
\end{align*}
Therefore,
$
T\var_C.((X_1^j,X_2^{j}+F_2^{j})\circ\Phi_j)
= T\var_C.T\var_C^{-1}.((X_1^j,X_2^{j})\circ\Phi_j)
= (X_1^j,X_2^{j})\circ\Phi_j
$
and we obtain the conditions \eqref{e:proofMC1} (which is equivalent to \eqref{e:mu_C^M}, \eqref{e:I_C} and \eqref{e:A_C}) 
and \eqref{e:MCthm1}, \eqref{e:MCthm2}. 

If $S$ can be chosen so that $F^C = 0$, then \eqref{e:Feom1KK} is equivalent to the Hamiltonian system $X_{H_C}$.  

The condition~\eqref{e:inv-cond} is rather strong and not satisfied for the case of Section~\ref{sec:SD}. Nevertheless, formula~\eqref{e:MCvarC} coincides with \eqref{e:SD-varC}. 
\end{remark}

\begin{remark}
The expression~\eqref{e:A_C} should be compared with Assumption~$M1$ in \cite{BLM01a}.
\end{remark}

\begin{remark}
The property $F_C|J^{-1}(0) = 0$ is relevant, because it means that the restriction to $J^{-1}(0)$ allows to reduce the controlled equations to a Hamiltonian system.  
\end{remark}

\section{Feedback control of mechanical systems on semi-direct products}\label{sec:SD}

\subsection{Setup}\label{sec:SD-setup}
Assume that $M$ and $G$ are Lie groups and that there is a right representation $\rho: M\to \textup{Aut}(G)$. 
That is $\rho^{\psi\phi} = \rho^{\phi}\circ\rho^{\psi}$ for $\phi,\psi\in M$. 
Let 
\[
 P = M\circledS G
\]
denote the semi-direct product. 
The right multiplication in $P$ is given by 
\begin{equation}
    \label{e:R_rho}
    R_{\rho}^{(\phi,g)}(\psi,h)
    = (\psi \phi, \rho^{\phi}(h)g ) . 
\end{equation}
The action by $G$ on $M\circledS G = P$ given by right multiplication (on the second factor) makes 
\[
 G\hookto P \to M
\]
into a (right) principal bundle.
The induced representation of $M$ on $\gu$ is again denoted by $\rho: M\to\textup{Aut}(\gu)$. The contragredient representation is 
\[
  \rho_*^{\phi}
  = (\rho^{\phi^{-1}})^*,
\]
and similarly we have 
$\rho^u X = \dd{t}{}|_0\rho^{\textup{exp}(t u)}X$ and 
$\rho_*^u p 
= \dd{t}{}|_0\rho_*^{\textup{exp}(t u)} p
= -(\rho^u)^*p
$.

Let $\mu_0^M$ be a right invariant metric on $M$. In the right trivialization this is given by an isomorphism $\mu_0^M: \mo\to\mo^*$, where $\mo$ is the Lie algebra of $M$ and $\mo^*$ its dual. 
Let $\I_0: \gu\to\gu^*$ a symmetric isomorphism such that 
\begin{align}
\label{e:I_0-inv}
    \Ad(g)^*\I_0\Ad(g) &= \I_0 
    = \rho_*^{\phi^{-1}}\I_0\rho^{\phi}
\end{align}
for all $(\phi,g)\in P$. 
Let $\A_0: TP\to\gu$ be a principal bundle connection. In the right trivialization this is given by 
\begin{equation}
    \label{e:SD-MC}
    \mathcal{A}_0: M\times G \times\mo\times\gu\to\gu,
    \quad
    (\phi,g,u,X) \mapsto \Ad(g^{-1})\rho^{\phi}(X + A_0(u))
\end{equation}
where $A_0: \mo\to\gu$ is independent of $(\phi,g)$.
The corresponding horizontal bundle is
\begin{equation}
    \hor_0
    =
    P\times\{(u,-A_0(u)): u\in\mo\}. 
\end{equation}
Let 
\begin{align*}
 \mu_0^P &= \mu^{KK}(\mu_0^M,\I_0,\A_0)
 =
 \left(
 \begin{matrix}
  \mu_0^M + A_0^*\I_0 A_0  &  A_0^*\I_0\\
  \I_0 A_0  &   \I_0 
 \end{matrix}
 \right)
 :
 \mo\times\gu\to\mo^*\times\gu^*
\end{align*}
be the (right invariant) Kaluza-Klein metric on $P$ corresponding to $(\mu_0^M,\I_0,\A_0)$. This metric has the properties
\begin{align*}
    \mu_0^P|\hor_0 = \mu_0^M\circ(T\pi,T\pi),\quad
    \mu_0^P|\ver = \I\circ(\A,\A),\quad 
    \mu_0^P(\hor_0,\ver) = 0.
\end{align*}


The Hamiltonian $H_0: P\times\po^*\to\R$ is 
\begin{equation}
    \label{e:SD-Ham}
    H_0(\om,\Pi)
    = \by{1}{2}\vv<\Pi,(\mu_0^P)^{-1}\Pi>.
\end{equation}
The momentum map with respect to the cotangent lifted $G$-action on $T^*P$ is denoted by $J: P\times\po^* \to\gu^*$. 
It is given by 
\begin{equation}
    \label{e:SD-momap}
    J(\om,\Pi)
    = \Ad(g)^*\rho_*^{\phi}(q)
\end{equation}
where $\om = (\phi,g)$ and $\Pi = (\nu,q)$.
The connection dependent isomorphism $\Psi_0 = \Psi_0^{\textup{hor}}\oplus\Psi_0^{\textup{ver}}: \W = P\times_M T^*M \oplus P\times\gu^* \to T^*P$ is, in the right trivialization, given by 
\[
 \Psi_0^{\textup{hor}}: 
 M\times G \times \mo^*\to P\times\po^*,\quad
 (\phi,g,\nu)\mapsto (\phi,g,\nu,0)
\]
and
\[
 \Psi_0^{\textup{ver}}: 
 M\times G\times \gu^* \to P\times\po^*
 ,\quad
 (\phi,g,p)
 \mapsto
 (\phi,g,
 A_0^*\rho_*^{\phi^{-1}}\Ad(g^{-1})^*p,
 \rho_*^{\phi^{-1}}\Ad(g^{-1})^*p). 
\]

\subsection{Control connection and symmetry actuation}\label{sec:SD-Gamma}
The control connection $\Gamma$ on the $G$-bundle  $P\to M$ is the direct product connection 
given by 
\begin{equation}
    \label{e:SD-Gamma}
    \A_{\Gamma}: 
    M\times G\times\mo\times\gu\to\gu,
    \quad
    (\phi,g,u,X)
    \mapsto
    \Ad(g^{-1})\rho^{\phi}X.
\end{equation}
Thus $\hor_{\Gamma} = P\times\mo\times\set{0}$. 
 
Let $\tilde{\jmath}: \hor^* = P\times\mo^*\to\gu^*$ be a $P$-equivariant and fiber-linear map. That is, 
\[
\tilde{\jmath}(\phi,g,\nu) = \Ad(g)^*\rho_*^{\phi}C\nu
\]
for a linear map
\begin{equation}
    \label{e:SD-C}
    C: \mo^*\to\gu^*.
\end{equation}
Since $\hpr_{\Gamma}^* = \pr_1$, definition~\eqref{e:B-j-tilde} becomes
\begin{align}
    \label{e:SD-B}
    B_{\om}q
    &=
    q + 
    \Ad(g)^*\rho_*^{\phi}C A_0^*\Big((\Ad(g)^*\rho_*^{\phi})^{-1}q\Big)\\
    &=
    \Ad(g)^*\rho_*^{\phi}
    \Big(
    1 + C A_0^*
    \Big)
    \rho_*^{\phi^{-1}}\Ad(g^{-1})^*q
    .\notag
\end{align}
The map $j: P\times\po^*\to\gu^*$ in Definition~\ref{def:SymCloLoop} is now
\begin{align}
\label{e:SD-j}
    j(\om,\Pi) 
    &= j(\phi,g,\nu,p)
    =
    B_{\om}^{-1}\Ad(g)^*\rho_*^{\phi}C \nu
    + (B_{\om}^{-1}-1)\Ad(g)^*\rho_*^{\phi}q \\
    &=
    \Ad(g)^*\rho_*^{\phi}\Big(
        (1+C A_0^*)^{-1} C \nu
        - C A_0^*(1 + C A_0^*)^{-1}q
    \Big).
    \notag
\end{align}
Further, the map $D_{\om}^0: \gu^*\to\mo^*$, defined in \eqref{e:D^0}, is given by 
\begin{equation}
    \label{e:SD-D^0}
    D^0_{\om}
    = 
    \Big(
        \pr_2.(\hl^{\mathcal{A}_0}).\hpr_{\Gamma}^*.\Psi_0^{\textup{ver}}
    \Big)_{\om}
    =
    A_0^*\rho_*^{\phi^{-1}}\Ad(g^{-1})^*
\end{equation}
where $\om = (\phi,g)$.

\subsection{Matching}
The strategy is to use item (2) of Theorem~\ref{thm:MC} to find the isomorphism $\Phi_C = (\var_C,S)$ and the force $F_C$. More precisely, $(\mu_C^M,\I_C,\A_C)$ are given by \eqref{e:mu_C^M}, \eqref{e:I_C}, \eqref{e:A_C}, and $(\var_C,S)$ and $F_2^C$ shall be determined from \eqref{e:thm_MC_2} such that \eqref{e:thm1} and \eqref{e:thm2} hold. 

Since the controlled system should again be right invariant, $\Phi_C$ should be $P$-equivariant (not only $G$-equivariant).
It follows that $\var_C = (\id_P,\var_C): P\times_M T^*M\to P\times_M T^*M = M\times G\times\mo^*$ for a $P$-independent map $\var_C: \mo^*\to\mo^*$ (again denoted by the same symbol).

The right-trivialization further implies that $X^j$ of Proposition~\ref{prop:X^j} can be expressed as $X^j = (X^j_1, X^j_2, X^j_3) \in \po\times\mo^*\times\gu^*$. 

Because the operator $\mu_0^M: \mo\to\mo^*$ corresponds to a right invariant metric on $M$, the Riemannian horizontal lift $\hl_{\mu_0^M}: TM\oplus TM\to TTM$ can be expressed as (e.g., \cite{Michor06})
\begin{equation}
    \label{e:SD-hl}
    \hl_{\mu_0^M}: \mo\times\mo\to\mo,
    \quad
    (u,w)
    \mapsto
    -\by{1}{2}\Big(\ad(u)^{\top}w + \ad(w)^{\top}u - \ad(w)u\Big)
\end{equation}
with dual
\begin{equation}
    \label{e:SD-hl*}
    \hl_{\mu_0^M}^*: \mo^*\times\mo\to\mo^*,
    \quad
    (\nu,w)
    \mapsto
    -\by{1}{2}
    \Big(\ad((\mu_0^M)^{-1}\nu)^*w + \ad(w)^*\nu - \ad(w)((\mu_0^M)^{-1}\nu) \Big)
\end{equation}
where the last equation follows since the connector \eqref{app:K} satisfies, by definition \eqref{app:hor}, $K^*\circ T\mu_0^M = \mu_0^M\circ K$. Hence, $u=(\mu_0^M)^{-1}\nu$ implies $(\hl_{\mu_0^M}^*)_{\nu}(u) = -\ad(u)^*\nu$.  
 
Fix $(\om,\nu,q) = (\phi,g,\nu,q)\in\W$ and let $(\om,\nu,p) = \Phi_j(\om,\nu,q)$. That is,
\begin{equation}
    \label{e:SD-p}
    p
    = q  - j\Psi_0^{\textup{hor}}(\om,\nu)
    = q - \Ad(g)^*\rho_*^{\phi}(1+CA_0^*)^{-1}C\nu .
\end{equation}
Proposition~\ref{prop:X^j} yields
\begin{equation}
    \label{e:SD-X_1^j}
    X_1^j(\om,\nu,p)
    = 
    \Big( (\mu_0^M)^{-1}\nu, -A_0u + \I_0^{-1}\rho_*^{\phi^{-1}}\Ad(g^{-1})^*p
    \Big)
    =:
    \Big( u, Y \Big)
    \in\po
    = \mo\times\gu 
\end{equation}
and 
\begin{equation}
    \label{e:SD-X_2^j}
    X_2^j(\om,\nu,p)
    = 
    -\ad(u)^*\nu
    - \vv<\rho_*^{\phi^{-1}}\Ad(g^{-1})^*p,
        \Curv_e^{\mathcal{A}_0}(u)>
    \in\mo^*.
\end{equation}
We use the diamond notation $\vv<X\diamond p,u> = \vv<p,\rho^u X>$ and, with $\tilde{p} = \rho_*^{\phi^{-1}}\Ad(g^{-1})^*p$, express the curvature term as
\begin{align}
    \label{e:SD-curvA}
    \vv<\tilde{p},
    \Curv_e^{\mathcal{A}_0}(u)>
    &=
    \ad(u)^*A_0^*\tilde{p}
    + A_0^*\ad(A_0u)^*\tilde{p}
    - A_0u\diamond \tilde{p}
    - A_0^*\rho_*^{u}\tilde{p}. 
\end{align}
Further, again by Proposition~\ref{prop:X^j} and equivariance of $j\circ\Psi_0^{\textup{hor}}$,
\begin{align}
\notag
    F_2^j(\om,\nu,p)
    &=
    D_{\om}^0B_{\om}d(j\circ\Psi_0^{\textup{hor}})(X_1^j,X_2^j)(\om,\nu,p) \\
    &=
    A_0^*(1+CA_0^*)\Big(
        (\ad(Y)^*+\rho_*^u)(1+CA_0^*)^{-1}C\nu 
        + (1+CA_0^*)^{-1}C X_2^j(\om,\nu,p)
        \Big)
        \label{e:SD-F_2^j-CX}\\
    &=
    A_0^*CX_2^j(\om,\nu,p)
    + 
    A_0^*(1+CA_0^*)
        (\ad(Y)^*+\rho_*^u)(1+CA_0^*)^{-1}C\nu .
        \notag
\end{align}
Notice that $\ad(A_0u)^*\tilde{p} = -\ad(\I_0^{-1}\tilde{p} - A_0u)^*\tilde{p} = -\ad(Y)^*\tilde{p}$.
With $\tilde{q} = \rho_*^{\phi^{-1}}\Ad(g^{-1})^*q$, equation~\eqref{e:SD-p} may be written as $(1+CA_0^*)^{-1}C\nu = \tilde{q} - \tilde{p}$.
Therefore,
\begin{align*}
    X_2^j(\om,\nu,p) + F_2^j(\om,\nu,p)
    &=
    X_2^j(\om,\nu,p)
    + A_0^*C X_2^j(\om,\nu,p)
    + (1+A_0^*C)A_0^*(\ad(Y)^*+\rho_*^u)(\tilde{q} - \tilde{p})\\
    &=
    (1+A_0^*C)\Big(
     -\ad(u)^*\nu
     -\ad(u)^*(A_0^*\tilde{p})
     + A_0 u\diamond \tilde{p}
     \Big) \\
    &\phantom{===}
    + (1+A_0^*C)A_0^*(\ad(Y)^*+\rho_*^u)\tilde{q}\\
    &=
    (1+A_0^*C)
    \Big(
        - \ad(u)^*((1+A_0^*C)^{-1}\nu)
        - \ad(u)^*(A_0^*\tilde{q})\\
        &\phantom{===}
        + A_0u\diamond \tilde{q}
        - A_0u\diamond (1+CA_0^*)^{-1}C\nu 
        + A_0^*(\ad(Y)^*+\rho_*^u)\tilde{q}
    \Big)\\
    &=
    (1+A_0^*C)
    \Big(
        - \ad(u)^*\mu 
        - \vv<\tilde{q},\Curv_e^{\mathcal{A}_0}(u)>
        - A_0u\diamond (1+CA_0^*)^{-1}C\nu 
    \Big)\\
    &=
    (1+A_0^*C)
    \Big(
        - \ad(u)^*\mu 
        - \vv<S\tilde{q},\Curv_e^{\mathcal{A}_C}(u)> \\
    &\phantom{===}
        + F_2^C(\om,\mu,S\tilde{q})
        - A_0u\diamond (1+CA_0^*)^{-1}C\nu 
    \Big)\\
    &=
    (1+A_0^*C)
    \Big(
    X_2^C(\om,\mu,S\tilde{q}) 
    + F_2^C(\om,\mu,S\tilde{q})
    - A_0u\diamond (1+CA_0^*)^{-1}C\nu 
    \Big)
\end{align*}
whence the isomorphism $\var_C: \mo^*\to\mo^*$ is
\begin{align}
\label{e:SD-varC}
    \mu
    &=
    \var_C(\nu)
    :=
    (1+A_0^*C)^{-1}\nu,
\end{align}
the force is
\begin{equation}
    \label{e:SD-force}
     F_2^C(\om,\mu,\tilde{q})
     := 
    \vv<\tilde{q},\Curv_e^{\mathcal{A}_C}(u)>
    -
    \vv<S^{-1}\tilde{q},\Curv_e^{\mathcal{A}_0}(u)>
\end{equation}
and the matching condition is
\begin{equation}
    \label{e:SD-MCond}
     A_0u\diamond (1+CA_0^*)^{-1}C\nu 
     \overset{!}{=} 0.
\end{equation}
Indeed, this condition holds, if and only if
\begin{equation}
\label{e:SD-MCond2}
    (X_2^j+F_2^j)\circ\Phi_j
    = 
    \var_C^*(X_2^C + F_2^C).
\end{equation}
Moreover, the force $F_C = (\Psi_C^{-1})^*(0,F_2^C,0)$ satisfies \eqref{e:thm1} and \eqref{e:thm2}, because \eqref{e:SD-force} vanishes for $q = 0$. 

If $C: \mo^*\to\gu^*$ satisfies \eqref{e:SD-MCond}, then the controlled Hamiltonian follows from the explicit formulas \eqref{e:mu_C^M}, \eqref{e:I_C} and \eqref{e:A_C} with \eqref{e:SD-varC}. 
Write the controlled connection $\A_C$ as $\A_C(e,e,u,X) = X + A_C(u)$. 
Therefore: 

\begin{theorem}\label{thm:SD-matching}
Let $\var_C$ be given by \eqref{e:SD-varC}, let $S: \gu^*\to\gu^*$ be a $\rho_*(M)\times\Ad(G)^*$-equivariant isomorphism
and let 
$\Phi = \Psi_C\circ(\id_P, \var_C, S)\circ\Phi_j^{-1}\circ\Psi_0^{-1}$.
Let
$H_C$ be the Kaluza-Klein Hamiltonian associated to 
\begin{align}
\label{e:SD-mu_C^M}
    \mu_C^M
    &= (1+A_0^*C)^{-1}\mu_0^M\\
    \label{e:SD-I_C}
    \I_C
    &= S\I_0 \\
    \label{e:SD-A_C}
    A_C 
    &= A_0 + \I_0^{-1} C(1+A_0^*C)^{-1}  \mu_0^M.
\end{align}
Then 
$\Phi^*(X_{H_C} + F_C) = X_{H_0} + F_0$
holds, if and only if 
\eqref{e:SD-MCond} is satisfied.
Moreover, $F_C|J^{-1}(0) = 0$.
If $S$ can be chosen such that $F_2^C = 0$, then $F_C=0$. 
\end{theorem}

\begin{remark}
Equations~\eqref{e:SD-mu_C^M} and \eqref{e:SD-A_C} have also been found in \cite{H19} by ad-hoc calculations. Now they follow from the general results in Theorem~\ref{thm:MC}, which also shows that the form of these equations is necessary, because of the requirements on $\Phi$ and the assumption that $H_C$ is of Kaluza-Klein form. 
\end{remark}

\begin{remark}
If the representation $\rho$ is trivial, so that $P = M\times G$ is a direct product of Lie groups, then the condition \eqref{e:SD-MCond} is empty. 
\end{remark}

\subsection{Equilibria}\label{sec:SD-equil}
Assume $\nu_e\in\mo^*$ is an equilibrium of the uncontrolled equation of motion in $M$. That is, 
\begin{equation}
\label{e:SD-equ1}
 \ad((\mu_0^M)^{-1}\nu_e).\nu_e
 = 0
\end{equation}
or, equivalently, $X_2^j(\om,\nu_e,0) = 0$ where $X_2^j$ is given in \eqref{e:SD-X_2^j}. Suppose that the matching condition \eqref{e:SD-MCond} holds and that $S$ has been found such that $F_C = 0$. 
Equation~\eqref{e:SD-MCond2} with $\om=(\phi,g)$ and $q=0$ implies that $\nu_e$ is also an equilibrium of
\begin{equation}
\label{e:SD-equ2}
 \ad((\mu_C^M)^{-1}\nu_e).\nu_e
 = 0
\end{equation}
if $(X_2^j+F_2^j)(\om,\nu_e,-j\Psi_0^{\textup{hor}}(\om,\nu_e)) = 0$. Because of equations~\eqref{e:SD-X_2^j} and \eqref{e:SD-F_2^j-CX}, a sufficient condition for this to hold is that 
\begin{equation}
    \label{e:SD-Cis0}
    C\nu_e = 0.
\end{equation}
Hence, if this condition holds, one can use (non-linear) stability analysis of \eqref{e:SD-equ2} for \eqref{e:SD-equ1}.

\section{Satellite with a rotor}\label{sec:SAT}
The feedback control of the satellite with a rotor is studied in \cite{BKMS92,BMS97,MarsdenPCL}. A Lie-Poisson version of the present approach (focusing on on the momentum shift factor $1-k$) is given in \cite{H19}. See also Remarks~1.1 and 1.3 in loc.\ cit.

\subsection{Setup}\label{sec:SAT-setup}
The configuration space of the satellite with a rotor attached to the third principal axis is a direct product Lie group $P = M\times G = \SO(3)\times S^1$. 
Let $I_1 > I_2 > I_3$ be the rigid body moments of inertia and $i_1 = i_2 > i_3$ those of the rotor. We use left multiplication in the direct product group $P$ to write the tangent bundle $TP \cong P\times\mathfrak{so}(3)\times\R \cong P\times\R^4$.
We identify $\mo = \so(3) = \R^3$.


Let  $\lam_j = I_j+i_j$ for $j=1,2,3$,
$\mu_0^M = \textup{diag}(\lam_1,\lam_2,I_3)$ and $\I_0 = i_3$.
The mechanical connection is 
\begin{equation}
    A_0: M\times\R^3\to \gu = \R,\;
    (\phi,\Om) \mapsto \Om_3.
\end{equation}
The associated (local) curvature form is denoted by $\CurvA = dA_0$. This is a two-form on $M$ and is given by 
\begin{equation}\label{e:K0}
    i(\Omega)\CurvA 
    =
    (
    \Om_2,
    -\Om_1,
    0
    ).
\end{equation}
The Kaluza-Klein metric associated to $(\mu_0^M,\I_0,A_0)$ is 
\begin{equation}
 \mu^P_0
 =
    \left(\begin{matrix}
    \lam_1 & & & \\
     & \lam_2 & & \\
     & & \lam_3 & i_3 \\
     & & i_3 & i_3
    \end{matrix}\right)
\end{equation}
with corresponding kinetic energy Hamiltonian $H_0$. 
The Hamiltonian system $X_{H_0}$ is invariant under the $S^1$-action on the second factor with momentum map $J: T^*P = P\times\R^3\times\R\to\R$, $(\om,\Pi,p)\mapsto p$. (It is also invariant under the left multiplication in $\SO(3)$, but this symmetry is not needed at this point.)

\subsection{The control connection}\label{sec:SAT-Gamma}
The control connection $\Gamma$ is 
\begin{equation}
    \label{e:SAT-Gamma}
    \Gamma:  
    \SO(3)\times S^1\times\R^3\times\R\to\R,
    \quad
    (\phi,g,\Om,X)
    \mapsto
    X.
\end{equation}
Thus $\hor_{\Gamma} = P\times\R^3\times\set{0}$. 
The maps $j$ and $B$ are defined by the formulas \eqref{e:SD-j} and \eqref{e:SD-B}. Hence the symmetry actuating force $F_0$ is given by \eqref{e:defSymAct} and $J+j$ is a conserved quantity for solutions of $X_{H_0}+F_0$. 

\subsection{Matching}\label{sec:SAT-matching}
Since $P$ is now a direct product, the matching condition \eqref{e:SD-MCond} is void and any linear map 
$
 \tilde{\jmath}_e = C: \mo^*=\R^3 \to \gu^* =\R
$
is admissible in the sense of Theorem~\ref{thm:SD-matching}. 
We choose
\begin{equation}
    \label{e:SAT-C}
    C(\Pi_1,\Pi_2,\Pi_3)
    = -k\Pi_3
\end{equation}
for a (small) parameter $k = -\gamma$. 

\begin{remark}
Because of equation~\eqref{e:SD-Cis0}, it follows that the corresponding force $F_0$ is potentially suitable to stabilize the unstable equilibrium $\Pi_e = (0,1,0)$ (rotation of the satellite about the middle axis). 
Indeed, \cite[Prop.~3.1]{BMS97} show that this control stabilizes $\Pi_e$ for $1>k>1-I_3/\lam_2$.
\end{remark}

Theorem~\ref{thm:SD-matching} yields the controlled data
\begin{align}
    \mu_C^M
    &=
    \left(\begin{matrix}
    \lam_1 & &  \\
     & \lam_2 & \\
     & & (1-k)^{-1}I_3 
    \end{matrix}\right) 
    \\
    \I_C 
    &=
    f_k^{-1}\I_0\\
    A_C
    &=
    f_k A_0
\end{align}
where
\begin{equation}
\label{e:SAT-f_k}
    f_k
    := 
    \frac{\mathbb{I}_0-k\lam_3}{\mathbb{I}_0(1-k)}
\end{equation}
follows from \eqref{e:SD-A_C} and the choice for $\I_C$ is such that $F_2^C$ in \eqref{e:SD-force} vanishes.

Let $H_C: T^*P\to\R$ be the Kaluza-Klein Hamiltonian specified by these formulas. Then we have shown that 
\begin{equation}
    \label{e:SAT-equiv}
    \Phi^*X_{H_C}
    = X_{H_0} + F_0
\end{equation}
where 
$\Phi 
= 
\Psi_C\circ(\id,\var_C,f_k^{-1})\circ\Phi_j^{-1}\circ\Psi_0^{-1}$. 

That is, the closed-loop equations associated to the conserved quantity $(J+j)(\om,\Pi,p)$, that is due to exerting the force $F_0$, are equivalent to the Hamiltonian system $X_{H_C}$. Explicitly, with Definition~\ref{def:SymCloLoop},
\begin{equation}\label{e:SAT-shift}
    (J+j)(\om,\Pi,p)
    = 
    (1+C\hpr_{\Gamma}^*)^{-1}(C\Pi + p)
    = 
    (1-k)^{-1}(-k\Pi_3 + p)
\end{equation}

\begin{remark}\label{rem:SATrem1}
Equation~\eqref{e:SAT-f_k} coincides with \cite[(1.27)]{H19}.
\end{remark}

\begin{remark}\label{rem:SATrem2}
The $1-k$ a posteriori shift that appears in \cite{BKMS92,BMS97,MarsdenPCL} has been included in the general construction and therefore is now a part of the formula \eqref{e:SAT-shift} for the conserved quantity. 
The shift is explicitly mentioned in the sentence below \cite[Equ.~(3.7)]{BMS97}.
The above construction is also different from \cite[Section~4]{H19}, where another relation between $A_0$ and $C$ was given. Both relations lead to the conclusion \eqref{e:SAT-equiv}, albeit with different formulas for $\I_C$. The ad-hoc approach of \cite[Section~4]{H19} corresponds to setting $B=1$ in Definition~\ref{def:SymCloLoop}. However, this means that equation \eqref{e:j-ver-0} need no longer hold. Since this equation is crucial when dealing with more general group actions, such as non-abelian $G$ or semi-direct $P$, we prefer definition \eqref{e:B-j-tilde} for $B$.
See Section~\ref{sec:CON-B=1}. 
\end{remark}

\begin{remark}\label{rem:SATrem3}
Remark~1.1 in \cite{H19} can now be rephrased as follows: 
The controlled equation of motion corresponding to the feedback law \eqref{e:SAT-shift} are obtained by replacing the uncontrolled equations of motion
$\dot{\nu} 
 = \ad((\mu_0^M)^{-1}\nu)^*\nu 
 - \vv<p, i((\mu_0^M)^{-1}\nu)\Curv^{A_0}>$
with
$\dot{\nu} 
 = \ad((\mu_C^M)^{-1}\nu)^*\nu 
 - \vv<p_k, i((\mu_C^M)^{-1}\nu)\Curv^{A_C}>$, 
where $p_k = (J+j)(\om,\Pi,p)$. 
\end{remark}


\section{Control of ideal fluids subject to an external Yang-Mills field}\label{sec:YM}

This section is concerned with ideal incomressible flow of charged particles in the presence of an external Yang-Mills field $\M = \textup{Curv}^{A_0} = dA_0 + \by{1}{2}[A_0,A_0]$.
Gauge symmetry yields conservation of charge along the fluid flow. 
In the language  of the previous sections, the symmetry actuating force is assumed to work on the charge variables, but not on the (positions or momenta of the) fluid particles. 
If the symmetry group $K$ is abelian, then $\M = dA_0$ is a magnetic field. 
See \cite{GBTV13,GHK83} for further background and \cite{H19} for the Lie-Poisson version of this example.

\subsection{Setup}\label{sec:YM-setup}
Let $K$ be a finite dimensional Lie group with Lie algebra $\ko$ and $Q$ a compact domain, possibly with boundary,  in $\R^n$. Consider the trivial principal bundle $S := Q\times K \to Q$ where the principal bundle action is given by right multiplication in the group. Let $\mu_0^Q = \vv<.,.>$ denote the induced Euclidean metric on $Q\subset\R^n$ and $\I_0$ a symmetric positive definite bilinear form on $\ko$ which is $\Ad(k)$-invariant. We fix a connection form $A_0: TQ\to\ko$. The horizontal space $\textup{Hor}_0\subset TS = Q\times K\times\R^n\times\ko$ is thus $\textup{Hor}_0 = \set{(u,k,u_x,X): X+A_0(x)u_x = 0}$.
Denote the  Kaluza-Klein metric on $Q\times K$ associated to $(\mu_0^Q,\I_0,A_0)$ 
by
\begin{equation}
    \label{e:mu0}
    \mu_0^S = \mu^{KK}(\mu_0^Q,\I_0,A_0).
\end{equation}
For $k\in K$, let $r^k: S\to S$, $(x,g)\mapsto(x,gk)$ denote the principal right action. Consider the volume preserving automorphisms  
\[
 \textup{Aut}_0(S)
 := \set{\Phi\in\textup{Diff}(S):
 \Phi\circ r^k = r^k\circ\Phi \; \forall k\in K \And \Phi^*\textup{vol}_S=\textup{vol}_S} 
\]
which can be identified as 
\[
 \Aut_0(S) = \textup{Diff}_0(Q)\circledS\F(Q,K) 
 \]
where $\textup{Diff}_0(Q)$ is the set of $\textup{vol}_Q$-preserving diffeomorphisms and
$\F(Q,K)$ denotes functions from $Q$ to $K$ (of a fixed differentiability class which we do not specify). See  \cite{GBTV13}.
Composition from the right gives rise to a right  representation 
\[
 \rho: \Diff_0(S) \to \textup{Aut}(\F(Q,K)),\quad
 \phi\mapsto \rho^{\phi}
\]
where $\rho^{\phi}(g) = g\circ\phi$. 
 
Consider the action by point-wise right multiplication $R^g: \F(Q,K)\to\F(Q,K)$, $h\mapsto hg$.
The induced semi-direct right action on $\Aut_0(S)$ is again denoted by $R$:
\begin{equation}\label{e:s-d}
 R^{(\psi,g)}(\phi,h) = (\phi\cdot\psi,R^g(\rho^{\psi}(h))). 
\end{equation}
In particular, $\Aut_0(S)\to\Aut_0(S)/\F(Q,K) = \Diff_0(S)$ is a right principal $\F(Q,K)$ bundle. 
We also consider the right trivializations
\begin{align}
    \label{e:triv}
    T\Aut_0(S) &
    \cong \Aut_0(S)\times\ao 
    = \Diff_0(S)\times\mathfrak{d}
        \times\F(Q,K)\times\gu
    \\
    \notag
    T^*\Aut_0(S) 
    &\cong \Aut_0(S)\times\ao^* 
    = \Diff_0(S)\times\mathfrak{d}^*
        \times\F(Q,K)\times\gu^* 
\end{align}
using the right multiplication in $\Aut_0(S)$, where $\ao = T_e\Aut_0(S)$, 
$\mathfrak{d} = T_e\Diff_0(S) = \X_0(Q)$ are divergence free vector fields tangent to the boundary  
and $\gu=T_e\F(Q,K) = \F(Q,\ko)$. Here, $\ao^*$, $\mathfrak{d}^*$ and $\gu^*$ denote the smooth part of the dual.

We have $\mathfrak{d}^* = \Om^1(Q)/d\mathcal{F}(Q)$ and $\gu^* = \F(Q,\ko^*)$. The pairings are given by 
\begin{align*}
    &\mathfrak{d}^*\times\mathfrak{d}\to\R,\quad ([\Pi],u)\mapsto\int_Q\vv<\Pi_x,u_x>\,dx\\
    &\gu^*\times\gu\to\R,\quad (q,X)\mapsto\int_Q\vv<q_x,X_x>\,dx
\end{align*}
where $[\Pi]$ is the class of $\Pi\in\Om^1(Q)$.
By definition, the smooth duals are the isomorphic images of the maps
\[
 [\mu_0^Q]: \mathfrak{d}\to\mathfrak{d}^*,\quad
  u \mapsto [\mu_0^Q(u)]
\]
and 
$
 \I_0: \gu\to\gu^*
$,
$q\mapsto\vv<\I_0 q,\_>$
where $[\mu_0^Q(u)]$ is the class of $\mu_0^Q(u)\in\Om^1(Q)$ in $\Om^1(Q)/d\mathcal{F}(Q) = \mathfrak{d}^*$ and $\I_0$ is independent of $q\in Q$. The inverse is of $[\mu_0^Q]$ is
\[
 [\mu_0^Q]^{-1}: [\Pi]\mapsto \mathcal{P}((\mu_0^Q)^{-1}(\Pi))
\]
where $\Pi$ is a representative of $[\Pi]$ and $\mathcal{P}$ is the Helmholtz-Hodge-Leray projection.
For a vector field $u\in\X(Q)$ the Helmholtz-Hodge-Leray projection is divergence free, tangent to the boundary, and given by $\mathcal{P}(u) = u - \nabla g$ where $g$ is determined by $\Delta g = \textup{div}\,u$ with Neumann boundary conditions.


In this section all duals are in the smooth sence such that $[\mu_0^Q]$ and $\I_0$ are isomorphisms. Further, the dual to $A_0: \doo\to\gu$ is $[A_0^*]: \gu^*\to\doo^*$, $q\mapsto[A_0^* q]$, where $A_0^*: \ko^*\to T^*Q$ is the point-wise dual and $[A_0^*q]\in\gu^*$   is the class of $A_0^*q\in\Om^1(Q)$ in $\Om^1(Q)/d\F(Q)$. 

The representation $\rho$ gives rise to an infinitesimal representations
$\rho^{\phi}X$ and 
$\rho^u(X) = dX.u = L_u X = \nabla_u X$ with $\phi\in\Diff_0(S)$, $u\in\mathfrak{d}$ and $X\in\gu$. The corresponding coadjoint representations are given by $\rho^{\phi}(q) = (\rho^{\phi^{-1}})^*(q)$ and $\rho^u(q) = (\rho^{-u})^*(q)$ with $q\in\gu^*$.

We define the bracket $[.,.]$ on $\mathfrak{d}$ (and similarly for $\ao$) to be the negative of the usual Lie bracket:
$
 [u,u] := -\nabla_{u}v + \nabla_{v}u
$
where $\nabla_{u}v=\vv<u,\nabla>v$ and $u,v\in\mathfrak{d}$. This choice of sign is compatible with \cite{A66,AK98}. 
Further, we define the operator $\ad(u).v = [u,v]$. Its dual is $\ad(u)^*[\Pi] = [\Pi\circ\ad(u)]$ for $[\Pi]\in\mathfrak{d}^*$. 

\begin{remark}[From finite to infinite dimensions]
\label{rem:YM-infdim}
The roles of the groups $P, M, G$ of Section~\ref{sec:SD} are now played, respectively, by $\Aut_0(S), \Diff_0(Q), \F(Q,K)$. 
In the above sections all spaces were assumed to be finite dimensional. However, all  calculations and formulas were global and carry over, in principle, to the infinite dimensional case. The only exception to this rule are the (fundamental) equations is equation~\eqref{app:OmK}, which is based on the local coordinate result \cite[Prop.~2.1]{HR15}. Nevertheless, the formulas of Proposition~\ref{prop:X^j} can be used for the Lie-Poisson case in this section, since \eqref{app:OmK} follows directly from the explicit expression of the symplectic form on $T^*\Aut_0(S)$ (see, e.g., \cite{Michor06}). 
Therefore, and also because this Lie-Poisson example has already been treated in \cite{H19}, where  further details on the setup are given, the algebraic formulas of Theorem~\ref{thm:SD-matching} will be applied without further justification. 
\end{remark}

\subsection{The Hamiltonian system}
Let $[\mu_0^S]: \ao\to\ao^*$ be the isomorphism associated to \eqref{e:mu0}. Equip $T^*\Aut_0^*(S) = \Aut_0(S)\times\ao^* = \Diff_0(Q)\times\F(Q,K)\times\doo^*\times\gu^*$ with the canoncal symplectic form and consider the Hamiltonian
\begin{equation}
    \label{e:YM-Ham0}
    H_0: \Aut_0(S)\times\ao^*\to\R,\quad
    (\om, \Pi)\mapsto\by{1}{2}\vv<\Pi, [\mu_0^S]^{-1}\Pi> .
\end{equation}
The metric $[\mu_0^S]$ is right-invariant with respect to \eqref{e:R_rho} and of Kaluza-Klein form with respect to $([\mu_0^Q],\I_0,\A_0)$. The mechanical connection is given by 
\begin{equation}
    \label{e:YM-A_0}
    \A_0: 
    \Diff_0(Q)\times\F(Q,K)\times\doo\times\gu\to\gu^*,\quad
    (\phi,g,u,X)
    \mapsto
    \Ad(g^{-1})\rho^{\phi}(X+A_0 u).
\end{equation}
This Hamiltonian is right-invariant. Hence the momentum map $J_0: T^*\Aut_0(S)\to\ao^*$ with respect to the cotangent lifted action is constant along solutions of $X_{H_0}$.  
In particular, $J  := \pr_2\circ J_0: T^*\Aut_0(S)\to\gu^*$
is a conserved quantity. This momentum map corresponds to the conservation of charge.

\subsection{The control connection}
The control connection $\Gamma$ is 
\begin{equation}
    \label{e:YM-Gamma}
    \Gamma: 
    \Diff_0(Q)\times\F(Q,K)\times\doo\times\gu\to\gu,\quad
    (\phi,g,u,X)
    \mapsto
    \Ad(g^{-1})\rho^{\phi} X.
\end{equation}
That is $\hor_{\Gamma} = \Aut_0(S)\times\doo\times\set{0}$. 
The maps $j$ and $B$ are defined by the formulas \eqref{e:SD-j} and \eqref{e:SD-B}. Hence the symmetry actuating force $F_0$ is given by \eqref{e:defSymAct} and $J+j$ is a conserved quantity for solutions of $X_{H_0}+F_0$.

\subsection{Matching condition}\label{sec:YM-MCond}
We apply Theorem~\ref{thm:SD-matching}. 
The matching condition \eqref{e:SD-MCond} is satisfied, if 
\begin{equation}
    \label{e:YM-MCond}
    (1 + C A_0^*)^{-1}C [\mu_0^Q] u
    = 
    \gamma \I_0 A_0 u
\end{equation}
for a parameter $\gamma$.
This equation is equivalent to
\begin{equation}
    \label{e:YM-C}
    C 
    = 
    \gamma \Big(
        1 - \gamma \I_0 A_0 [\mu_0^Q]^{-1}[A_0^*]
    \Big)^{-1}
    \I_0 A_0 [\mu_0^Q]^{-1}
\end{equation}
where $\gamma$ is assumed to be sufficiently small so that this expression exists. 
Due to \eqref{e:SD-A_C} this 
implies that the controlled connection must be given by
\begin{equation}
    \label{e:YM-A_C}
    A_C = (1+\gamma) A_0.  
\end{equation}
Therefore, $\Curv^{\mathcal{A}_C} = (1+\gamma)\Curv^{\mathcal{A}_0}$ and choosing \eqref{e:SD-I_C} as
\begin{equation}
    \label{e:YM-I_C}
    \I_C = (1+\gamma)^{-1}\I_0
\end{equation}
yields a force $F_C$, which vanishes because of definition \eqref{e:SD-force}.

These formulas were found in \cite[Section~2.G]{H19}, involving an ad-hoc analysis of the force field and the corresponding closed-loop equations. 

Let $H_C: T^*\Aut_0(S)\to\R$ be the Kaluza-Klein Hamiltonian specified by Theorem~\ref{thm:SD-matching}.
That is,
$[\mu_C^Q] = (1+[A_0^*]C)^{-1}[\mu_0^Q]$ and $\I_C$, $A_C$ are given above.
Then we have shown that 
\begin{equation}
    \label{e:YM-equiv}
    \Phi^*X_{H_C}
    = X_{H_0} + F_0
\end{equation}
where 
$\Phi 
= 
\Psi_C\circ(\id,\var_C,(1+\gamma)^{-1})\circ\Phi_j^{-1}\circ\Psi_0^{-1}$. 

That is, the closed-loop equations associated to the conserved quantity $J+j$, that is due to exerting the force $F_0$, are $\Phi$-related to the Hamiltonian system $X_{H_C}$.

\section{Conclusions and future directions}\label{sec:CON}

\subsection{Setting $B=1$ and non-abelian $G$}\label{sec:CON-B=1}
Instead of \eqref{e:B-j-tilde}, one may also set $B=1$ in Definition~\ref{def:SymCloLoop}. 
While one then loses the simplifying equation~\eqref{e:j-ver-0}, the controlled Noether Theorem~\ref{thm:ConsLaw} remains valid since it is independent of the choice of $B$. Moreover, as mentioned in Remark~\ref{rem:SATrem2}, $B=1$ does work for the satellite example. 

To see why \eqref{e:B-j-tilde} is preferable, consider the direct product $P = M\times G$ of Lie groups, where $G$ is non-abelian. As in Section~\ref{sec:SD}, the strategy is to look for a fiber-linear isomorphism 
\[
 \Phi_C
 =
 (\var_C,S): 
 P\times\mo^*\times\gu^*
 \to 
 P\times\mo^*\times\gu^*
\]
where $\var_C: P\times\mo^*\to P\times\mo^*$, $(\om,\nu)\mapsto(\om,\var_C(\om,\nu))$ is equivariant for the right action and $S: \gu^*\to\gu^*$ is $\Ad^*(G)$-equivariant, such that 
\begin{equation}
    T\var_C^{-1}.
    \Big(
     (X_1^C, X_2^C+F_2^C)(\Phi_C(\om,\nu,q))
    \Big)
    =
    \Big(
     X_1^j, X_2^j+F_2^j
    \Big)
    \Big(\Phi_j(\om,\nu,q)\Big)
\end{equation}
where $(\om,\nu,q) = (\phi,g,\nu,q)\in\W$. Equivariance of $\var_C$ implies that $\var_C(\om,\nu) = \var_C(\nu)$ is independent of $\om$. 
Let us assume that $S$ has been found so that $F_2^C = 0$.
Hence we obtain the necessary matching condition 
\begin{equation}
    \label{e:CON-MC}
    T\var_C^{-1}.
     X_2^C(\Phi_C(\om,\nu,q))
    =
    \Big(
     X_2^j+F_2^j
    \Big)
    \Big(\Phi_j(\om,\nu,q)\Big)
\end{equation}
Now, because of \eqref{app:X2} and linearity of $S$ and $\var_C$, the left hand side of this equation has to depend linearly on $q$.
To evaluate this equation at 
$(\om,\nu,p) = \Phi_j(\om,\nu,q) = (\phi, g ,\nu, q - j\Psi_0^{\textup{hor}}(\om,\nu))$, rewrite \eqref{e:X_2^j} as 
\[
 X_2^j(\om,\nu,p)
 = -\ad(u)^*\nu - L(\om,\nu,p),
\]
where $L$ is defined in \eqref{e:L}, and \eqref{e:F_2^j} as
\begin{align*}
    F_2^j(\om,\nu,p)
    &=
    D_{\om}^0.d(j\circ\Psi_j).X^j(\om,\nu,p) \\
    &=
    A_0^*\Big(
        \ad(X)^*C\nu
        + [\ad(X)^*,C A_0^*(1+CA_0^*)^{-1} ]_{\textup{op}}\Ad(g^{-1})^*p\\
    &\phantom{===}
        + C X_2^j(\om,\nu,p)
        + CA_0^*(1+CA_0^*)^{-1}\Ad(g^{-1})^*X_3^j(\om,\nu,p)
    \Big)
\end{align*}
where $D_{\om}^0 = A_0^*\Ad(g^{-1})^*$, $[.,.]_{\textup{op}}$ denotes the commutator of operators and the equations
\[
 j\Psi_0^{\textup{hor}}(\om,\nu)
 = \Ad(g)^*C\nu 
\]
and 
\[
 j\psi_j^{\textup{ver}}(\om, p)
 =
 \Ad(g)^* C A_0^* (1+C A_0^* )^{-1} \Ad(g^{-1})^* p
\]
are used. Further, $u = (\mu_0^M)^{-1}\nu$ and 
$X 
= -A_0u + \Ad(g)\tilde{\I}^{-1}p
= -A_0 u + \I_0^{-1}(1+CA_0^*)^{-1}\Ad(g^{-1})^*p$,
because of \eqref{e:IC}.  
Now, \eqref{e:X_3^j} equates to
\[
 (1+CA_0^*)^{-1}\Ad(g^{-1})^*X_3^j(\om,\nu,p)
 = 
 [\ad(X)^*,C A_0^*(1+CA_0^*)^{-1} ]_{\textup{op}}\Ad(g^{-1})^*p
\]
and the first term in \eqref{e:L}, which defines $L(\om,\nu,p)$, becomes
\begin{align*}
    \vv<&
     d(j\circ\Psi_j^{\textup{ver}})(\hl_{u_2}^{A_0}(\om), 0), \tilde{\I}^{-1}p 
     > \\
    &=
    -\vv<\ad(A_0 u_2)^* CA_0^*(1+CA_0^*)^{-1}\Ad(g^{-1})^*p ,
     \I_0^{-1}(1+CA_0^*)^{-1}\Ad(g^{-1})^*p> \\
    &\phantom{=}
    +\vv<CA_0^*(1+CA_0^*)^{-1}\ad(A_0 u_2)^*\Ad(g^{-1})^*p , \I_0^{-1}(1+CA_0^*)^{-1}\Ad(g^{-1})^*p> \\
    &= 
    \Big\langle 
     A_0^*\Big(
       \ad\Big( \I_0^{-1}(1+CA_0^*)^{-1}\Ad(g^{-1})^*p \Big)^*
       \Big(CA_0^*(1+CA_0^*)^{-1}\Ad(g^{-1})^*p \Big)\\
    &\phantom{=}
       -
        \ad\Big( (CA_0^*(1+CA_0^*)^{-1})^*\I_0^{-1}(1+CA_0^*)^{-1}\Ad(g^{-1})^*p \Big)^*    
        \Big( CA_0^*(1+CA_0^*)^{-1}\Ad(g^{-1})^*p \Big)
     \Big) 
     , u_2 
    \Big\rangle
\end{align*}
Therefore, collecting all the terms in the right hand side of \eqref{e:CON-MC}, which are quadratic in $p$, and equating these to $0$, yields the condition
\begin{align}
    \label{e:CON-MC2}
    &(1+CA_0^*)A_0^*
    \ad\Big( (CA_0^*(1+CA_0^*)^{-1})^*\I_0^{-1}(1+CA_0^*)^{-1}\Ad(g^{-1})^*p \Big)^*
    \Big( CA_0^*(1+CA_0^*)^{-1}\Ad(g^{-1})^*p \Big)\\
    &=
    A_0^*CA_0^*
    \ad\Big( \I_0^{-1}(1+CA_0^*)^{-1}\Ad(g^{-1})^*p \Big)^*\Big(                    \Ad(g^{-1})^*p \Big).
 \notag
\end{align}
In order for \eqref{e:CON-MC} to be linear in $q$, and therefore in $p$, this equality should hold, whence we obtain a new necessary (but not sufficient) matching condition. 
Since the $\ad(\cdot)^*$-operation is in $\gu^*$, this condition is automatically satisfied when $\gu$ is abelian. 
We have not shown that matching is impossible for non-abelian $G$ under the assumption $B=1$, but with \eqref{e:CON-MC2} there is (at least) one new condition on $C$ that has to be satisfied. On the other hand, Theorem~\ref{thm:SD-matching} shows that there are \emph{no conditions} on $C: \mo^*\to\gu^*$ for the \emph{direct product} Lie group $P = M\times G$, regardless of whether $G$ is commutative or not.  

Consistently with this conclusion and Remark~\ref{rem:SATrem2}, for the satellite example the symmetry group $G$ is the abelian group $S^1$.

\subsection{Comparison with existing literature}\label{sec:CON-COMP}
The method of controlled Lagrangian and Hamiltonian systems has started with \cite{Kri85,BKMS92} and then been further developed in \cite{BLM01a, BLM01b, CBLMW02, BMS97, PB19}. A review of these results is contained in \cite{BL02}. In \cite{CBLMW02} it is shown, for general systems and without explicitly considering symmetries, that the methods of controlled Lagrangian and Hamiltonian systems are equivalent.
The matching result of \cite{BLM01a} is more general than Theorem~\ref{thm:SD-matching}, since it allows the metric $\mu_0^M$ to depend on the base point. On the other hand, \cite{BLM01a, BLM01b, BL02, BKMS92, BMS97, PB19} assume the symmetry group $G$ to be abelian and, when $M$ is a Lie group, the product $M\times G$ is assumed to be a direct product Lie group.

With regard to the existing literature, cited in the previous paragraph, the approach of the present paper differs in the following ways:
\begin{enumerate}
    \item 
    The symmetry group $G$ is allowed to be non-abelian. 
    \item
    The product $P$ may be a semi-direct product with respect to a representation $\rho: M\to\Aut(G)$. 
    \item
    The starting point for the construction is the \emph{symmetry actuating force} defined in \eqref{e:defSymAct} with $B$ given in \eqref{e:B-j-tilde}. This force leads to a conserved quantity (Theorem~\ref{e:ConsLaw}), and Theorem~\ref{thm:MC} provides explicit formulas for the controlled data $(\mu_C^M,\I_C,A_C)$ together with a non-explicit matching condition \eqref{e:thm_MC_2} under completely general assumptions. 
    In the case where $P$ is a (semi-)direct product of Lie groups $M$ and $G$, the matching condition can be made explicit and is given in Theorem~\ref{thm:SD-matching}. We emphasize that the matching condition is automatically satisfied, if the representation $\rho: M \to\Aut(G)$ is trivial.  
    In contrast, the approach of \cite{BLM01a, BLM01b, BL02, BKMS92, CBLMW02, BMS97, PB19} is to start from $(\mu_C^M,\I_C,A_C)$ and to construct the force from the controlled data. 
    While these approaches are, in principle, equivalent, the advantage of starting directly with the force is that the corresponding control law is automatically related to the associated conserved quantity $J+j$. 
    Otherwise, the control law could be related to $J+j$ only up to a factor. While in some cases (such as the satellite example), this factor can be absorbed by multiplying the conserved momentum variable accordingly, for more general group actions (compare Section~\ref{sec:CON-B=1}) such an absorption may not be possible.  
    An example of this relation is given in Remarks~\ref{rem:SATrem2} and \ref{rem:SATrem3}. 
    \item
    The crucial step in the construction of the force $F$ in \eqref{e:defSymAct} is the choice \eqref{e:B-j-tilde}. This choice implies, that $F$ depends on the full $\Pi$-dynamics. Let us, locally, decompose $T^*P = M\times G\times\mo^*\times\gu^*$ and correspondingly $\Pi = (\om,g,\nu,p)$. Since $F$ is $G$-invariant, we thus write, again locally, $F = F(\om,\nu,p)$. The point is that $F$ depends also on observations of the controlled $p$-variable. This is in contrast to \cite[Prop.~3.1]{BLM01a} where the Simplified Matching Assumptions, which are used in the examples, imply that the control force \cite[Equ.~(3.12)]{BLM01a} is independent of the controlled symmetry (position and velocity) variables. Similarly, the control in \cite{BLM01b} is chosen so that the force \cite[Equ.~(24)]{BLM01b} is $p$-independent.
    \item
    All of the calculations in this paper are global and can therefore be adapted to the infinite dimensional case relevant for fluid mechanics. The only exception to this rule is \eqref{app:OmK}, which depends on the local coordinate calculation in \cite[Prop.~2.1]{HR15}. However, for the infinite dimensional case, that is relevant for fluid dynamics, this equation can be shown directly.  See, e.g., \cite[Section~4.1]{Michor06}.
\end{enumerate}

\subsection{Electromagnetic flow control}\label{sec:CON-EMF} 
The Hamiltonian formulation of inviscid flow of charged particles in interaction with electromagnetic or, more generally, Yang-Mills fields is given in \cite{GHK83, MWRSS83, GR08}. 
For the incompressible case the phase space of this system can be schematically described as
\[
 T^*(\Aut_0\times\textup{Conn})
\]
where $\Aut_0$ is the volume preserving automorphism group of a principal fiber bundle and $\textup{Conn}$ is the space of connections. In contrast to Section~\ref{sec:YM} there are now two (commuting) symmetry group actions: the action by right multiplication by $\Aut_0$ remains; additionally there is now the diagonal action due to left multiplication by the gauge group. 

This example still fits the general setup of Section~\ref{sec:match} and one may envisage a construction similar to that of Section~\ref{sec:SD} in order to transform \eqref{e:thm_MC_2} into an explicit matching criterion.

\subsection{Closed-loop control of stochastic Hamiltonian systems}\label{sec:CON-SHC}
The force in the feedback control construction of \eqref{e:Feom1} depends on the Hamiltonian itself, the momentum map $J: T^*P\to\gu^*$ and a $G$-equivariant fiber-linear map $\tilde{\jmath}: \hor^* = J^{-1}(0)\to\gu^*$. Hence a generalization to stochastic Hamiltonian systems could be possible.

More concretely, consider a $G$-invariant Hamiltonian $H_0: T^*P\to\R$ as in  Section~\ref{sec:SYMACT}.
For $k=1,\dots,N$, let $H_k: T^*P\to\R$ be a collection of $G$-invariant Hamiltonians and $W^k$ pair-wise independent Brownian motions. Consider the stochastic Hamiltonian system
\[
 \delta\Pi_t
 = X_{H_0}(\Pi_t)\,\delta t
  + \sum_{k=1}^N X_{H_k}(\Pi_t)\,\delta W_t^k
\]
where $\delta$ is the Stratonovich differential. The symmetry actuating force $F$ depends on the system variable $\Pi$ and is therefore subject to the same uncertainty as the system itself. That is, with equation~\eqref{e:CLU}: 
\[
 F(\Pi_t)
 = - \vl^*\Big(\Pi_t,
        \Big((\tau,J)|\ver_{\Gamma}^*\Big)^{-1}
            (\tau(\Pi_t) , B.dj.\delta\,\Pi_t )
        \Big).
\]
For $k=0,\dots,N$, we thus define 
\[
 \mathcal{U}_k(\Pi)
 =
 \Big((\tau,J)|\ver_{\Gamma}^*\Big)^{-1}
 \Big(
  \tau(\Pi),-B_{\om}dj.X_{H_k}(\Pi)
 \Big)
\]
and, analogously to \eqref{e:CLU},
\[
 F(\Pi)
 = 
 \vl^*\Big(\Pi,
  \mathcal{U}_0(\Pi)\,\delta t
  + \sum_{k=1}^N \mathcal{U}_k(\Pi)\,\delta W_t^k
 \Big)
 =
 \vl^*\Big(\Pi,
  \mathcal{U}_0(\Pi) \Big)
  \,\delta t
  + \sum_{k=1}^N 
   \vl^*\Big(\Pi,\mathcal{U}_k(\Pi)
    \Big)\,\delta W_t^k
\]
which should be viewed as a Stratonovich operator. The stochastic controlled Hamiltonian system is 
\begin{equation}
    \label{e:STOCH-CHS}
    \delta \Pi_t
    =
    X_{H_0}(\Pi_t)\,\delta t
    + \sum_{k=1}^N X_{H_i}(\Pi_t)\,\delta W_t^k
    + F(\Pi_t).
\end{equation}
As with \eqref{e:ConsLaw},
it can be verified that the stochastic controlled Noether Theorem holds: 
\[
 d\Big(J+j\Big).
 \Big( X_{H_0}(\Pi_t)\,\delta t
    + \sum_{k=1}^N X_{H_i}(\Pi_t)\,\delta W_t^k
    + F(\Pi_t)\Big) = 0,
\]
whence $(J+j)(\Pi_t)$ is constant in $t$. 

It is therefore conceivable, that the construction in the above sections would also work for stochastic Hamiltonian systems, as long as the stochastic perturbations respect the level sets of the momentum map $J$ (which is the case if the $H_k$ are $G$-invariant). 
One may consider perturbations in the $M$-variables (e.g.\ satellite or fluid), in the $G$-variables (e.g.\ rotor or external magnetic field), or in both. 

The modern formulation of stochastic geometric mechanics has been initiated in \cite{LCO08} and \eqref{e:STOCH-CHS} is to be understood in the Stratonovich formulation of loc.\ cit.  Examples of finite dimensional stochastic Hamiltonian (and almost-Hamiltonian) systems are in \cite{LCO08,H13,HR15,ACH16} and examples of stochastic Hamiltonian fluid dynamical systems can be found in \cite{Holm15,CFH17,H18,DHL19}. 
The stochastic energy-Casimir method, whose deterministic version has been successfully applied to stabilize closed-loop systems in \cite{BLM01a, BLM01b, BL02, BKMS92, BMS97, H19}, has been developed by \cite{AGH18}.



\section{Appendix}\label{sec:app}
Let $G\hookto P \to M$ be a right principal bundle with a connection form $\A\in\Om^1(P,\gu)$. Suppose $\mu^M$ is a Riemannian metric on $M$ and $\I$ a locked inertia tensor on $\gu$. Let $\mu^P = \mu^{KK}(\mu^M,\I,A)$ be the corresponding Kaluza-Klein metric on $P$ and $H: T^*P\to\R$, $\Pi\mapsto\by{1}{2}\vv<\Pi,(\mu^P)^{-1}\Pi>$ the kinetic energy Hamiltonian. The standard momentum map is $J: T^*P\to\gu^*$. 

Let 
\begin{equation}
    \label{app:W}
    \mathcal{W} 
    := P\times_M T^*M \oplus P\times\gu^*
\end{equation}
and we shall identify $P\times_M T^*M \oplus P\times\gu^* = P\times_M T^*M \times\gu^*$.
Consider the $\A$-dependent isomorphism
\begin{equation}
\label{app:Psi}
    \Psi: 
    \mathcal{W}\to T^*P,\quad
    (\om, \nu, q)
    \mapsto 
    (\hl^*)^{-1}(\om,\nu)
    +  (J|\ver_{\mathcal{A}}^*)_{\om}^{-1}(q)
\end{equation}
where $\hl^*: \hor^* \to P\times_M T^*M$ is the adjoint to the $\A$-horizontal lift, $\om\in P$ and $\nu\in T^*M$ with $\pi(\om) = \tau_{T^*M}(\nu)$. 
See \cite{W78}. 

Denote the canonical Liouville and symplectic forms on $T^*P$ by $\theta$ and $\Om = -d\theta$, respectively. Then the pulled back symplectic form $\Om^{\mathcal{A}} := \Psi^*\Om = -d(\Psi^*\theta)$ can be expressed as
\begin{equation}
    \label{app:OmA}
    \Om^{\mathcal{A}}
    =
    \pi_2^*\Om^{T^*M} 
    - \vv<dJ_{\mathcal{W}} \,\overset{\wedge}{,}\, \tau^*\mathcal{A}>
    - \vv<J_{\mathcal{W}}, \tau^*\textup{Curv}^{\mathcal{A}}>
    + \vv<J_{\mathcal{W}}, \by{1}{2}\tau^*[\mathcal{A},\mathcal{A}]_{\wedge} >
\end{equation}
where $\pi_2: \mathcal{W}\to T^*M$, $\Om^{T^*M}$ is the canonical symplectic form on $T^*M$, $\tau: \mathcal{W}\to P$ is the projection, $\textup{Curv}^{\mathcal{A}} = d\mathcal{A} + \by{1}{2}[\mathcal{A},\mathcal{A}]_{\wedge}$ is the curvature form and $J_{\mathcal{W}}: \mathcal{W}\to\gu^*$, $(\om,\nu,q)\mapsto q$ is the momentum map of the induced $G$-action on $\mathcal{W}$. 
Further, we use the general notation $\vv<A\,\overset{\wedge}{,}\,B> = \vv<A,B> - \vv<B,A>$ and $[\mathcal{A},\mathcal{A}]_{\wedge}(U,V) = 2[\mathcal{A}(U),\mathcal{A}(V)]$. (See \cite[Section~19.2]{Mdgb} for a definition of the graded bracket $[\phantom{x}, \phantom{x}]_{\wedge}$.)

Let $\tau_{TM}: TM \to M$ and $\tau_{T^*M}: T^*M\to M$ be the foot point projections. The corresponding vertical spaces are $\ver(\tau_{TM}) = \ker T\tau_{TM}\subset TTM$ and $\ver(\tau_{T^*M}) = \ker T\tau_{T^*M}\subset TT^*M$. The Riemannian metric $\mu^M$ defines a connection on $TM\to M$ and the corresponding horizontal space will be denoted by $\hor(\mu^M)\subset TTM$. This induces a connection on $T^*M\to M$ with horizontal space $\hor^*(\mu^M) = T\mu^M(\hor(\mu^M))$. Hence, there are splittings
\begin{equation}
    \label{app:hor}
    TTM = \hor(\mu^M)\oplus\ver(\tau_{T^*M})
    \textup{ and }
    TT^*M = \hor^*(\mu^M)\oplus\ver(\tau_{T^*M}).
\end{equation}
The vertical lift maps give rise to the isomorphisms $\textup{vl}: TM\oplus TM \cong \ver(\tau_{TM})$ and $\textup{vl}^*: T^*M\oplus T^*M \cong \ver(\tau_{T^*M})$. This allows to define the connector maps
\begin{align}
    \label{app:K}
    K &:= \pr_2\circ\textup{vl}^{-1}\circ\vpr(\mu^M):
    TTM\to TM\\
    \notag
    K^* &:= \pr_2\circ(\textup{vl}^*)^{-1}\circ\vpr^*(\mu^M):
    TT^*M\to T^*M
\end{align}
where $\pr_2$ are the projections onto the second factor and $\vpr(\mu^M)$, $\vpr^*(\mu^M)$ are the projections onto $\ver(\mu^M)$, $\ver^*(\mu^M)$ respectively. See \cite[Section~22.8]{Mdgb}.

\begin{proposition}
The canonical symplectic form on $T^*M$ can be expressed as
\begin{equation}
    \label{app:OmK}
    \Om^{T^*M}
    = \vv<T\tau_{T^*M}\,\overset{\wedge}{,}\, K^*>.
\end{equation}
\end{proposition}

This equation is the global version of the 
``$\Om = \sum dq^i\wedge dp_i$  formula''. 
I do not know of a reference (other than the related equation in \cite[Prop~2.1]{HR15}), thus a proof is included: 

\begin{proof}
In \cite[Prop.~2.1]{HR15} it is shown that 
\[
 (\mu^M)^*\Om^{T^*M} 
 = \Om^{T^*M}\circ\Lam^2 T\mu^M
 = \vv<T\tau_{TM}\,\overset{\wedge}{,}\, \mu^M\circ K>
\]
whence
\begin{align*}
 \Om^{T^*M}
 &= \vv<T\tau_{TM}\,\overset{\wedge}{,}\, K>\circ\Lam^2(T\mu^M)^{-1}
 = \vv<T\tau_{TM}\circ(T\mu^M)^{-1}\,\overset{\wedge}{,}\, \mu^M\circ K \circ(T\mu^M)^{-1}>\\
 &= \vv<T\tau_{T^*M}\,\overset{\wedge}{,}\, K^*>
\end{align*}
where we use that $K^* =  \mu^M\circ K \circ(T\mu^M)^{-1}$.
\end{proof}

Elements $W\in\mathcal{W}$ are written as $W=(\om,\nu,q)\in P\times T^*M\times \gu^*$ where it is understood that $\pi(\om) = \tau_{T^*M}(\nu)$.
The pulled back Hamiltonian is 
\begin{equation}
    \Psi^*H: \mathcal{W}\to\R,\quad
    (\om,\nu,q)\mapsto \by{1}{2}\vv<\nu,(\mu^M)^{-1}\nu>
        + \by{1}{2}\vv<q,\I_{\om}^{-1}q>.  
\end{equation}
Let $Y_2 = (\dot{\om}_2,\dot{\nu}_2,0) \in T_{(\om,\nu,q)}\mathcal{W}$ such that $K^*(\dot{\nu}_2) = 0$ and $\mathcal{A}(\dot{\om}_2) = 0$. The fibered product requirement yields additionally $T\pi(\dot{\om}_2) = T\tau_{T^*M}(\dot{\nu}_2) = \dot{x} \in T_x M$ where $x = \pi(\om)$. We define the horizontal derivative $d^{\textup{hor}}(\Psi^*H)_{(\om,\nu,q)} \in T_x^*M$ by
\[
 \vv< d^{\textup{hor}}(\Psi^*H)_{(\om,\nu,q)}, \dot{x}_2>
 = 
 \vv<d(\Psi^*H)_{(\om,\nu,q)}, Y_2>.
\]
If $X^{\mathcal{A}} = (\Om^{\mathcal{A}})^{-1}d(\Psi^*H)$ is the Hamiltonian vector field with components $(X^{\mathcal{A}}_1,X^{\mathcal{A}}_2,X^{\mathcal{A}}_3) = X^{\mathcal{A}}(\om,\nu,q)$, then it follows from \eqref{app:OmK} that these are given by
\begin{align}
\label{app:X1}
    X^{\mathcal{A}}_1
    &=
    \textup{hl}^{\mathcal{A}}_{\om}\Big((\mu^M)^{-1}\nu\Big)
    + \zeta_{\mathbb{I}_{\om}^{-1}q}(\om) \\
\label{app:X2}
    X^{\mathcal{A}}_2
    &=
    (\textup{hl}_{\mu^M}^*)_{\nu}\Big((\mu^M)^{-1}\nu\Big)
    +
    \textup{vl}_{\nu}^*\Big(
        - d^{\textup{hor}}(\Psi^*H)_{(\om,\nu,q)}
        - \vv<q, i((\mu^M)^{-1}\nu)\textup{Curv}^{\mathcal{A}}_x>
        \Big)\\
\label{app:X3}
    X^{\mathcal{A}}_3
    &= 0
\end{align}
where $\textup{hl}_{\mu^M}^*: T^*M\oplus TM \to \hor^*(\mu^M)\subset TT^*M$ is the horizontal lift,
that is 
$(\textup{hl}_{\mu^M}^*)_{\nu} = (T_{\nu}\tau_{T^*M}|\hor^*(\mu^M))^{-1}$,
%
%
and $\zeta_X\in\mathcal{P}$ is the fundamental vector field (infinitesimal generator) associated to the $G$-action and $X\in\gu$. Note that $T\pi.X^{\mathcal{A}}_1 = T\tau_{T^*M}.X^{\mathcal{A}}_2 = (\mu^M)^{-1}\nu$.

\end{document}